\documentclass[11pt,a4paper]{article}

\usepackage{latexsym}
\usepackage{amsfonts}
\usepackage{amssymb}
\usepackage{amsmath}
\usepackage{wasysym}
\usepackage[mathscr]{eucal}
\usepackage{theorem}
\usepackage{enumerate}
\usepackage{cite}
\usepackage{url}
\usepackage{graphicx}
\usepackage[all]{xy}

\theoremstyle{plain}
\theorembodyfont{\itshape}
\newtheorem{thm}{Theorem}

\newtheorem{prp}{Proposition}

\theorembodyfont{\upshape}

\theorembodyfont{\upshape}

\newcommand{\qed}{\hfill\mbox{\raggedright $\Box$}\medskip}

\newcommand{\mydate}{
 \ifcase\month \or
 January\or February\or March\or April\or May\or June\or
 July\or August\or September\or October\or November\or December\fi
 \space \number\year}

\newcommand{\smcirc}{{\scriptstyle \,\circ\,}}
\newcommand{\ssmcirc}{{\scriptscriptstyle \,\circ\,}}
\newcommand{\bwedge}{\raisebox{0.2ex}{${\textstyle \bigwedge}$}}
\newcommand{\bvee}{\raisebox{0.2ex}{${\textstyle \bigvee}$}}

\newcommand{\mathslf}[1]{\ensuremath{\mbox{\slshape\textsf{#1}}}}

\begin{document}

\title{Lie Groupoids in Classical Field Theory II: \\
       Gauge Theories, Minimal Coupling and Utiyama's Theorem \\ \mbox{}}

\author{Bruno T.\ Costa$^{\,1,2}$, Michael Forger$^{\,1}$~and~%
        Luiz Henrique P.\ P\^egas$^{\,1}$ \\ \mbox{}}
\date{\normalsize
      $^{1\,}$ Instituto de Matem\'atica e Estat\'{\i}stica, \\
      Universidade de S\~ao Paulo, \\
      Caixa Postal 66281, \\
      BR--05315-970~ S\~ao Paulo, SP, Brazil \\[4mm]
      $^2\,$ Departamento de Matem\'atica, \\
      Universidade Federal de Santa Catarina, \\
      Rua Jo\~ao Pessoa, 2750, \\
      BR--89036-256~ Blumenau, SC, Brazil
     }
\maketitle

\thispagestyle{empty}

\vspace*{5mm}

\begin{abstract}
\noindent
 In the two papers of this series, we initiate the development of a new
 approach to implementing the concept of symmetry in classical field theory,
 based on replacing Lie groups/algebras by Lie groupoids/algebroids, which
 are the appropriate mathematical tools to describe local \linebreak
 symmetries when gauge transformations are combined with space-time
 transformations. \linebreak
 In this second part, we shall adapt the formalism developed in the first
 paper to the context of gauge theories and deal with minimal coupling and
 Utiyama's theorem.
\end{abstract}

\vspace*{5mm}

\begin{flushright}
 \parbox{12em}{
  \begin{center}
   Universidade de S\~ao Paulo \\
   RT-MAP-1801 \\
   February 2018
  \end{center}
 }
\end{flushright}

\newpage

\setcounter{page}{1}

\section{Introduction}

In the first paper of this series~\cite{CFP}, we have initiated an
investigation of how to handle symmetries~-- or more precisely,
local symmetries~-- in classical field theories using the language
of Lie groupoids and their actions.
However, the formalism developed there is perhaps a bit too
general because it allows us to leave the nature of the underlying
Lie groupoids and their actions completely unspecified, whereas
there can be no doubt that the motivation for the entire program
comes predominantly from one single (class of) example(s),
namely, gauge theories.
Spelling out the details for this case is the main goal of the present
paper and is necessary not only because it provides us with a class
of examples whose importance can hardly be overestimated but
also because it leads to a substantial clarification of the general
structure of the theory.
Moreover, the results will generalize those of earlier work~\cite{FS}
by extending them from internal symmetries to space-time symmetries.

Let us begin with a few comments on the already traditional geometric
formulation of gauge theories (as classical field theories) over a general
space-time manifold~$M$; more details can be found in textbooks such 
as~\cite{Ble,GS,FF}.
The basic input data one has to fix right at the start are an internal
symmetry group, which is a Lie group $G_0^{}$ with Lie algebra
$\mathfrak{g}_0^{}$,%
\footnote{Note that we perform a slight change of notation as compared
to Ref.~\cite{FS}, where we have denoted the internal symmetry group
by~$G$ and its Lie algebra by~$\mathfrak{g}$: here, we want to reserve
these symbols for the basic Lie groupoid and Lie algebroid of the theory.}
together with a principal bundle~$P$ over~$M$ with structure group~$G_0^{}$
and bundle projection $\, \rho: P \longrightarrow M$: then gauge fields
are described in terms of connections in~$P$, which can be viewed as
sections of an affine bundle over~$M$, namely, the connection bundle~%
$\, CP = JP/G_0^{}$ of~$P$.
Moreover, if the theory is to contain not only gauge fields (as in
``pure'' Yang-Mills theories) but also matter fields, one also has
to fix a vector space~$V$ equipped with a representation of~$G_0^{}$
or, more generally, a manifold~$Q$ equipped with an action of~$G_0^{}$:
then matter fields are described by sections of the associated vector
bundle $\, E = P \times_{G_0} V$ (for scalar matter fields) or of its
tensor product with some tensor or spinor bundle over~$M$ (for
tensor or spinor matter fields) or of the associated fiber bundle
$\, E = P \times_{G_0} Q$ (for nonlinear scalar matter fields
such as in the nonlinear sigma models).
Finally, there is gravity, described by yet another and very
special kind of field, namely, a metric tensor $\mathslf{g}$
on~$M$. (Some discussion of what sets the metric tensor
apart from all other fields can be found in Ref.~\cite{HE}.)

Symmetries in this approach are traditionally described in terms
of automorphisms of the principal bundle~$P$ and the induced
automorphisms of its connection bundle and its associated bundles.
To set the stage, recall that an \emph{automorphism} of~$P$ is
a diffeomorphism of~$P$ as a manifold which is $G_0^{}$-%
equivariant, i.e., which commutes with the right action of the
structure group~$G_0$ on~$P$: since the orbits of this action
are precisely the fibers of~$P$, it then follows that it takes
points in the same fiber to points in the same fiber and hence
induces a diffeomorphism of the base manifold~$M$.
Moreover, the automorphism is said to be \emph{strict} if it
preserves the fibers, or equivalently, if the induced diffeo%
morphism on the base manifold is the identity.
Automorphisms of~$P$ form a group~$\mathrm{Aut}(P)$ and strict
automorphisms of~$P$ form a normal subgroup~$\mathrm{Aut}_s(P)$
which is the kernel of a natural group homomorphism
\[
 \mathrm{Aut}(P)~\longrightarrow~\mathrm{Diff}(M)
\]
that projects each automorphism of~$P$ to the diffeomorphism of~$M$
it induces.
In physics language, strict automorphisms are also called \emph{gauge
transformations} and the group $\mathrm{Aut}_s(P)$ is often called the
\emph{gauge group} and denoted by $\mathrm{Gau}(P)$, but we prefer
the more precise term \emph{group of gauge transformations} so as to
avoid the confusion whether by ``gauge group'' one means the infinite-%
dimensional group $\mathrm{Gau}(P)$ or the finite-dimensional structure
group~$G_0$.
Thus strict automorphisms, or gauge transformations, are \emph{internal
symmetries} since they do not move points in space-time, whereas general
automorphisms will in what follows be referred to as \emph{space-time
symmetries}.%
\footnote{There is some abuse of language in this simplified terminology
because general automorphisms always represent a mixture of ``pure''
space-time symmetries with internal symmetries.
The problem here is that there is in general no natural notion of a
``pure'' space-time symmetry, since that would require a \emph{lifting}
of the group $\mathrm{Diff}(M)$ (or at least of an appropriate subgroup
thereof) to realize it as a subgroup (and not only as a quotient group)
of~$\mathrm{Aut}(P)$, whose elements would then represent the
``pure'' space-time symmetries.
However, such a lifting may not even exist, and even if it does (which
happens, e.g., when the principal bundle~$P$ is trivial), it is far from
unique, so what one means by a ``pure'' space-time transformation
still depends on which lifting is chosen.}
At any rate, all such symmetry transformations, being represented by
automorphisms of~$P$, can be lifted to automorphisms of its jet bundle
$JP$ and hence act naturally on the connection bundle $\, CP = JP/G_0^{}$
of~$P$ as well as on any associated vector bundle or fiber bundle $E$,
its jet bundle~$JE$ and any tensor or spinor bundle over~$M$, thus
providing the appropriate setting for deciding which of them are
symmetries of the field theoretical model under consideration.

The main mathematical difficulty within this approach comes from the
fact that one is dealing here with infinite-dimensional groups which
are notoriously hard to handle from the point of view of Lie theory.
Therefore, it is desirable to recast the property of invariance of a
field theory under such local symmetries into a form where one deals
exclusively with finite-dimensional objects.
This program has been initiated in Ref.~\cite{FS} and implemented
there for strict automorphisms (gauge transformations), where it leads
naturally to replacing Lie groups by Lie group bundles (and similarly
Lie algebras by Lie algebra bundles), making use of the well-known
fact that there is a natural isomorphism between the group of strict
automorphisms of~$P$ and the \emph{group of sections} of the
\emph{gauge group bundle} of~$P$, which is the Lie group bundle
$P \times_{G_0} G_0^{}$ associated to~$P$ via the action of~%
$G_0^{}$ on itself by conjugation:
\[
 \mathrm{Aut}_s(P)~\cong~\Gamma(P \times_{G_0} G_0^{}) \,.
\]
In order to extend the resulting analysis from strict automorphisms to
general automorphisms, we have to go one step further and replace
Lie groups or Lie group bundles by Lie groupoids (and similarly Lie
algebras or Lie algebra bundles by Lie algebroids). In this case, the
basic observation is that there is a natural isomorphim between the
group of automorphisms of~$P$ and the \emph{group of bisections}
of the \emph{gauge groupoid} of~$P$, which is the Lie groupoid
$(P \times P)/G_0^{}$ obtained as the quotient of the cartesian
product of two copies of~$P$ by the ``diagonal'' right action
of~$G_0^{}$:
\[
 \mathrm{Aut}(P)~\cong~\mathrm{Bis}((P \times P)/G_0^{}) \,.
\]
Thus our task in what follows will be to extend the results of
Ref.~\cite{FS} by applying the general formalism of Ref.~\cite{CFP}
to this specific situation.

When we replace Lie groups by Lie groupoids, or to put it a bit more
precisely, actions of Lie groups on manifolds by actions of Lie groupoids
on fiber bundles (over the same base manifold), we have to face one
important novel feature, namely, that the construction of induced
actions will involve changing the Lie groupoid as well.
For example, while an action of a Lie group $G_0^{}$ on a manifold~$X$
induces an action of the same Lie group $G_0^{}$ on its tangent bundle~%
$TX$, an action of a Lie groupoid $G$ on a fiber bundle~$E$ (both over
the same base manifold~$M$) induces an action not of the original Lie
grupoid~$G$ but rather of its jet groupoid $JG$ on the jet bundle $JE$
of~$E$.
(A similar phenomenon already occurs for Lie group bundles, as observed
in Ref.~\cite{FS}.)
As it turns out, properly dealing with this feature is the key to make the
entire theory work out smoothly.

Let us pass to briefly describe the contents of the paper.
In Section~$2$, we present the minimal coupling prescription and the
curvature map that enters the formulation of Utiyama's theorem in a
very general context, and we show that these constructions are
invariant (or perhaps it might be better to say, equivariant) under
any action of any Lie groupoid over space-time on the bundle of
field configurations over space-time, provided we employ the correct
induced actions of the pertinent Lie groupoids derived from the former
on the pertinent bundles derived from the latter.
We conclude with a series of comments intended to show why, from
the point of view of field theory, this approach is excessively general
and needs to be adapted to a setting where all bundles are derived
from some principal bundle and all connections are derived from
principal connections in that principal bundle~-- which is the
standard setup for gauge theories anyway.
In Section~3, we collect the technical tools needed to perform this
adjustment and to state the main results.
The first step here is to recall the definition of the gauge groupoid~$G$
of a principal bundle~$P$ and of its natural actions on any bundle~$E$
associated to~$P$ (including $P$ itself).
Next, we introduce the (first order) jet groupoid $JG$ of~$G$ and use
the results of the previous section and of Ref.~\cite{CFP} to write down
natural actions of~$JG$ on various derived bundles such as the jet
bundle~$JP$ and the connection bundle~$CP$ of~$P$ or the jet
bundle~$JE$ of any bundle~$E$ associated to~$P$.
We also show how iterating this procedure provides induced actions
of the second order jet groupoid $J^{\>\!2} G$ and, more generally,
the semiholonomous second order jet groupoid $\bar{J}^{\>\!2} G$
of~$G$ on the semiholonomous second order jet bundle
$\bar{J}^{\>\!2} P$ and on the (first order) jet bundle
$J(CP)$ of the connection bundle~$CP$ of~$P$.
In Section~4, we then prove the main theorems concerning the
invariance (or perhaps it might be better to say, the equivariance)
of the minimal coupling prescription and the curvature map under
the actions of the pertinent Lie groupoids introduced in the previous
section, thus providing the desired extension of the results of Ref.~%
\cite{FS} from the setting of Lie group bundles (internal symmetries)
to that of Lie groupoids (space-time symmetries).

In an appendix, we present an interesting result that links some of
our constructions to analogous constructions using jet prolongations
of principal bundles.
This subject is treated in great generality in Ref.~\cite{KMS} at the
level of principal bundles and their associated bundles, but is not
addressed at the level of Lie groupoids; in fact, the concept of Lie
groupoid does not appear there at all.
The basic ingredient is the (first order) jet prolongation $P^{(1)}$
of the given principal bundle~$P$, which is a principal bundle over
the same base manifold and whose structure group $G_0^{(1)}$ is the
jet group of the structure group $G_0^{}$ of~$P$, as defined, e.g.,
in Ref.~\cite{KMS}.
This allows us not only to show that various bundles derived from a
bundle $P \times_{G_0} Q$ associated to~$P$ (such as its jet bundle
$J(P \times_{G_0} Q)$ and the tangent bundle $T(P \times_{G_0} Q)$
of its total space) or even just from~$P$ itself (such as its connection
bundle $\, CP = JP/G_0^{}$) are bundles associated to~$P^{(1)}$
(which is not new), but also that the jet groupoid $J((P \times P)/
G_0^{})$ of the gauge groupoid $(P \times P)/G_0^{}$ of~$P$
is canonically isomorphic to the gauge groupoid $(P^{(1)} \times
P^{(1)})/G_0^{(1)}$ of~$P^{(1)}$, or to put it more bluntly:
jet groupoids of gauge groupoids are gauge groupoids!
However, we have not explored all consequences of this approach,
since this is not needed to derive our results.

\section{Minimal coupling and Utiyama's theorem I}

As stated in the introduction, our main goal in this paper is to extend
the results of Ref.~\cite{FS} about invariance of the minimal coupling
prescription and of the curvature map (Utiyama's theorem) from the
context of Lie group bundles to that of Lie groupoids.
To do so, let us begin by recalling the general definition of these two
constructions.

The term ``minimal coupling'' is widely used in mathematical physics
to denote a procedure for converting ordinary derivatives to covariant
derivatives.
Such derivatives apply to ``matter fields'' on space-time~$M$ which in
a general geometric framework are sections of some fiber bundle~$E$
over~$M$: then their ordinary derivatives are sections of its (first order)
jet bundle~$JE$, as a fiber bundle over~$M$, while their covariant
derivatives are sections of its linearized (first order) jet bundle
\begin{equation} \label{eq:LJETB1}
 \vec{J} E~\cong~L(\pi^*(TM),VE)~\cong~\pi^*(T^* M) \otimes VE \,,
\end{equation}
as a fiber bundle over~$M$, where $\pi$ is the bundle projection
from~$E$ to~$M$, $\pi^*(TM)$ resp.\ $\pi^*(T^* M)$ is the pull-back
of the tangent resp.\ cotangent bundle of~$M$ to~$E$, $VE$ is the
vertical bundle of~$E$ and $L(\pi^*(TM),VE)$ denotes the bundle
of fiberwise linear maps from~$\pi^*(TM)$ to~$VE$.
Within this context, the minimal coupling prescription states that
the covariant derivative $D\varphi$ of a section $\varphi$ of~$E$ is
obtained from its ordinary derivative $\partial\varphi$ by using a
connection in~$E$ to decompose the tangent bundle $TE$ of (the
total space of)~$E$ into the direct sum of the vertical bundle $VE$
and horizontal bundle $HE$ and then projecting onto the vertical part.
Now if we think of that connection as being given by its horizontal
lifting map, which is a section $\varGamma$ of~$JE$ as an affine
bundle over~$E$, so that at each point $e \in E$ with $\pi(e) = x$,
$\varGamma(e)$ is a linear map from $T_x M$ to~$T_e E$ whose
image is the horizontal space $H_e E$ at~$e$ of the connection,
then that projection onto the vertical part is precisely
$1 - \varGamma(e) \smcirc T_e \pi$.
Thus if $\, \varphi \in \Gamma(M,E)$, so that $\, \partial \varphi \in
\Gamma(M,JE)$ \linebreak and $\, D\varphi \in \Gamma(M,\vec{J} E)$,
then as maps from $M$ to~$JE$, or equivalently, as fiberwise linear
maps from~$TM$ to~$TE$, $\partial\varphi$ is just the first order jet
(or tangent map) of~$\varphi$, while $D\varphi$ is the difference
\begin{equation} \label{eq:COVDER1}
 D\varphi~=~\partial\varphi - \varGamma \smcirc \varphi \,.
\end{equation}
This rule can be recast in a purely algebraic form, namely, by viewing
it as the result of inserting $\partial\varphi$ and $\varGamma \smcirc
\varphi$ into the \emph{difference map} for (first order) jet bundles,
i.e., the bundle map
\begin{equation} \label{eq:DIFMAP1}
 -: JE \times_E JE~~\longrightarrow~~
 L(\pi^*(TM),VE)~\cong~\pi^*(T^* M) \otimes VE
\end{equation}
over~$E$,  explicitly constructed as follows: given any point
$e \in E$ with $\pi(e) = x$ and any two jets $\, u_e^1,u_e^2
\in J_e^{} E \subset L(T_x^{} M,T_e^{} E)$, we have
$\, T_e^{} \pi \,\smcirc\, u_e^i = \mathrm{id}_{T_x M}^{}$,
for $i=1,2$, and hence the difference $u_e^1 - u_e^2$ (in the
vector space $L(T_x^{} M,T_e^{} E)$) takes values in the kernel
of $T_e^{} \pi$, that is, the vertical space $V_e E$ of~$E$, so
it becomes a linear map from $T_x M$ to $V_e E$.

The construction of the ``curvature map'' for connections in a given
fiber bundle $E$ over~$M$ is similar but somewhat more complicated
because it involves its semiholonomous second order jet bundle
$\bar{J}^{\>\!2} E$.
To see how that goes, we proceed as in Ref.~\cite{CFP} by first
constructing the iterated jet bundle $J(JE)$ of~$E$ and noting that
this allows two projections to~$JE$, namely, the iterated jet target
projection $\, \pi_{J(JE)}^{}: J(JE) \longrightarrow JE \,$ as well
as the jet prolongation $\, J\pi_{JE}^{}: J(JE) \longrightarrow JE$
\linebreak
of the jet target projection $\, \pi_{JE}^{}: JE \longrightarrow E \,$:
then by definition, $\bar{J}^{\>\!2} E$ is the subset of~$J(JE)$
where these two projections coincide.
Concretely, for $e \in E$, $u_e^{} \in J_e^{} E \,$ and $\, u'_{u_e}
\in J_{u_e}^{}(JE)$,
\begin{equation} \label{eq:SOJB1}
 (\pi_{J(JE)}^{})_{u_e}^{}(u'_{u_e})~=~u_e^{}~~,~~
 (J\pi_{JE}^{})_{u_e}^{}(u'_{u_e})~
 =~T_{u_e}^{} \pi_{JE}^{} \,\smcirc\, u'_{u_e} \,.
\end{equation}
As it turns out~\cite[Theorem~5.3.4, p.~174]{Sau}, $\bar{J}^{\>\!2} E$
is an affine bundle over~$JE$ which decomposes naturally into a symmetric
part and an antisymmetric part: the former is precisely the usual second
order jet bundle $J^{\>\!2} E$ of~$E$ (sometimes also called the holo%
nomous second order jet bundle of~$E$) and is an affine bundle over~$JE$,
with difference vector bundle equal to the pull-back to~$JE$ of the vector
bundle $\, \pi^* \bigl( \bvee^{\;\!2\,} T^\ast M \bigr) \otimes VE \,$
over~$E$ by the jet target projection $\pi_{JE}^{}$, whereas the latter
is a vector bundle over~$JE$, namely the pull-back to~$JE$ of the vector
bundle $\, \pi^* \bigl( \bwedge^{\!2\,} T^\ast M \bigr) \otimes VE \,$
over~$E$ by the jet target projection~$\pi_{JE}^{} \,$:
\begin{equation} \label{eq:SOJB2}
 \begin{array}{c}
  \bar{J}^{\>\!2} E~\cong~J^2 E \; \times_{JE}^{} \; \pi_{JE}^*
                          \Bigl( \pi^* \bigl( \bwedge^{\!2\,} T^\ast M \bigr)
                                            \otimes VE \Bigr) \,, \\[2mm]
  \vec{J^2} E~\cong~\pi_{JE}^*
                    \Bigl( \pi^* \bigl( \bvee^{\;\!2\,} T^\ast M \bigr)
                                      \otimes VE \Bigr) \,.
 \end{array}
\end{equation}
Now the proofs of these statements given in Ref.~\cite{Sau} and else%
where in the literature all involve local coordinate representations,
so it may be of some interest to provide a more direct, global argument.
To this end, consider what we shall call the \emph{difference map} for
semiholonomous second order jet bundles, i.e., the bundle map
\begin{equation} \label{eq:DIFMAP2}
 -: \bar{J}^{\>\!2} E \times_{JE} \bar{J}^{\>\!2} E~~
 \longrightarrow~~L^2(\pi^*(TM),VE)~\cong~
 \pi^* \bigl( \bigotimes\nolimits^{\!2} T^\ast M \bigr) \otimes VE
\end{equation}
over~$\pi_{JE}^{}$, where $L^2(\pi^*(TM),VE)$ denotes the bundle
of fiberwise bilinear maps from~$\pi^*(TM)$ to~$VE$, explicitly
constructed as follows: given any point $e \in E$ with $\pi(e) = x$,
any jet $u_e^{} \in J_e^{} E$ and any two semiholonomous second
order jets $\, u_{u_e}^{\prime\,1},u_{u_e}^{\prime\,2} \in
\bar{J}_{u_e}^{\>\!2} E \subset J_{u_e}(JE) \subset
L(T_x^{} M,T_{u_e}^{}(JE))$, we have $\, T_{u_e}^{} \pi_{JE}^{}
\,\smcirc\, u_{u_e}^{\prime\,i} = u_e^{}$, for $i=1,2$, and hence the
difference $u_{u_e}^{\prime\,1} - u_{u_e}^{\prime\,2}$ takes values
in the kernel of $T_{u_e}^{} \pi_{JE}^{}$, that is, the vertical space
$V_{u_e}^{\mathrm{jt}}(JE)$ of~$JE$ with respect to the jet target
projection~$\pi_{JE}$ from~$JE$ to~$E$.
But with respect to this projection, $JE$ is an affine bundle with difference
vector bundle $\vec{J} E$, so this vertical space is canonically isomorphic
to the corresponding difference vector space,
\[
 V_{u_e}^{\mathrm{jt}}(JE)~\cong~L(T_x^{} M,V_e^{} E) \,,
\]
and thus the difference $u_{u_e}^{\prime\,1} - u_{u_e}^{\prime\,2}$
becomes a linear map from $T_x M$ to this vector space, which can be
identified with a bilinear map from $T_x M$ to $V_e E$.
Obviously, any such bilinear map can be canonically decomposed into
its symmetric and its antisymmetric part, and the restriction of the
difference map for semiholonomous second order jet bundles to the
symmetric part will provide the \emph{difference map} for second
order jet bundles, i.e., the bundle map
\begin{equation} \label{eq:DIFMAP3}
 -: J^{\>\!2} E \times_{JE} J^{\>\!2} E~~
 \longrightarrow~~L_s^2(\pi^*(TM),VE)~\cong~
 \pi^* \bigl( \bvee^{\;\!2\,} T^\ast M \bigr) \otimes VE
\end{equation}
over~$\pi_{JE}^{}$, where $L_s^2(\pi^*(TM),VE)$ denotes the bundle
of fiberwise symmetric bilinear maps from~$\pi^*(TM)$ to~$VE$.
Moreover, it will provide an \emph{alternator} or \emph{antisymmetrizer}
for semiholonomous second order jets, which is an affine bundle map 
\begin{equation} \label{eq:ASPSSOJ}
 \mathrm{Alt} : \bar{J}^{\>\!2} E~~\longrightarrow~~
 L_a^2(\pi^*(TM),VE)~\cong~
 \pi^* \bigl( \bwedge^{\!2\,} T^\ast M \bigr) \otimes VE
\end{equation}
over~$\pi_{JE}^{}$, where $L_a^2(\pi^*(TM),VE)$ denotes the bundle
of fiberwise antisymmetric bilinear maps from~$\pi^*(TM)$ to~$VE$,
as follows: given any point $e \in E$ with $\pi(e) = x$, any jet
$u_e^{} \in J_e^{} E$ and any semiholonomous second order jet
$\, u'_{u_e} \in \bar{J}_{u_e}^{\>\!2} E$, choose any holonomous
second order jet $\, u_{u_e}^{\prime\,0} \in J_{u_e}^{\>\!2} E$ and define
$\mathrm{Alt}(u'_{u_e})$ to be the antisymmetric part of the
difference $u'_{u_e} - u_{u_e}^{\prime\,0}$, which obviously
does not depend on the choice of $u_{u_e}^{\prime\,0}$.
It is this construction that we shall use to define the curvature
of a connection in~$E$, given, say, in terms of its horizontal
lifting map, which is a section $\varGamma$ of~$JE$ as a
bundle over~$E$: observing that its jet prolongation $j\varGamma$
will then be a section not just of~$J(JE)$ but actually
of~$\bar{J}^{\>\!2} E$, again as a bundle over~$E$, since
$\, T\pi_{JE}^{} \,\smcirc\, j\varGamma = T \bigl( \pi_{JE}^{}
\,\smcirc\, \varGamma \bigr) = T \, \mathrm{id}_E^{} =
\mathrm{id}_{TE}^{}$, and noting that it will therefore be
a section of~$\bar{J}^{\>\!2} E$ along $\varGamma$ when
$\bar{J}^{\>\!2} E$ is considered as a bundle over~$JE$
instead, we can compose it with the alternator to produce a
section of $\, \pi_{JE}^* \bigl( \pi^* \bigl( \bwedge^{\!2\,}
T^\ast M \bigr) \otimes VE \bigr) \,$ along $\varGamma$,
which is just a section of $\pi^* \bigl( \bwedge^{\!2\,}
T^\ast M \bigr) \otimes VE$ and (possibly up to a sign which
is a matter of convention) is the curvature
\begin{equation} \label{eq:CURV1}
 \mathrm{curv}(\varGamma)~
 =~\mathrm{Alt} \,\smcirc\, j\varGamma
\end{equation}
of the given connection.

The main statement we want to prove in this section is that these
two constructions are invariant (or perhaps it might be better to say,
equivariant) under any action of any Lie groupoid $G$ over~$M$
on the bundle~$E$ over~$M$, provided we employ the correct
induced actions of the pertinent Lie groupoids derived from~$G$
on the pertinent bundles derived from~$E$.

Thus assume we are given a Lie groupoid~$G$ over~$M$, with
source projection $\, \sigma_G^{}: G \longrightarrow M$ \linebreak
and target projection $\, \tau_G^{}: G \longrightarrow M$, together
with an action
\begin{equation} \label{eq:ACTLG1}
 \begin{array}{cccc}
  \Phi_E^{}: & G \times_M E & \longrightarrow &     E    \\
             &    (g,e)     &   \longmapsto   & g \cdot e
 \end{array}
\end{equation}
of~$G$ on~$E$. (Cf.\ equation~(44) of Ref.~\cite{CFP}.)
Then we obtain an induced action
\begin{equation} \label{eq:IAGLJB1}
 \begin{array}{cccc}
  \Phi_{VE}^{}: & G \times_M VE & \longrightarrow &     VE     \\
                &    (g,v_e)    &   \longmapsto   & g \cdot v_e
 \end{array}
\end{equation}
of~$G$ on the vertical bundle $VE$ of~$E$, defined by
\begin{equation} \label{eq:IAGLJB2}
 g \cdot v_e^{}~=~T_e^{} L_g^{} (v_e^{}) \,,
\end{equation}
where $TL_g^{}$ denotes the tangent map to $L_g^{}$;
in other words, left translation by~$g$ in~$VE$ is just the
derivative of left translation by~$g$ in~$E$.
(Cf.\ equations~(89) and~(90) of Ref.~\cite{CFP}.)
Combining this with the natural action of the linear frame groupoid
$GL(TM)$ of the base manifold~$M$ on the cotangent bundle $T^* M$
of~$M$, we obtain an induced action of the Lie groupoid $\, GL(TM)
\times_M G \,$ on the linearized jet bundle $\vec{J} E$ of~$E$,
\begin{equation} \label{eq:IACT08}
 \begin{array}{ccc}
  \bigl( GL(TM) \times_M G \bigr) \times_M \vec{J} E
  & \longrightarrow & \vec{J} E \\[1mm]
             \bigl( (a,g),\vec{u}_e \bigr)          
  &   \longmapsto   & (a,g) \cdot \vec{u}_e
 \end{array}
\end{equation}
as suggested by the isomorphism of equation~(\ref{eq:LJETB1}), defined by
\begin{equation} \label{eq:IACT09}
 (a,g) \cdot \vec{u}_e~
 =~T_e L_g \,\smcirc\, \vec{u}_e \,\smcirc\, a^{-1} \,.
\end{equation}
(Cf.\ equations~(96) and~(98) of Ref.~\cite{CFP}.)
On the other hand, applying the jet functor to all structural maps that
appear in the original action~(\ref{eq:ACTLG1}), we obtain an induced
action
\begin{equation} \label{eq:IAJ1GJ1B1}
 \begin{array}{cccc}
  \Phi_{JE}: & JG \times_M JE & \longrightarrow &      JE       \\[1mm]
             &   (u_g,u_e)    &   \longmapsto   & u_g \cdot u_e
 \end{array}
\end{equation}
of the jet groupoid~$JG$ of~$G$ on the jet bundle~$JE$ of~$E$,
defined by
\begin{equation} \label{eq:IAJ1GJ1B2}
 u_g^{} \cdot u_e^{}~
 =~T_{(g,e)} \Phi_E^{} \,\smcirc\, (u_g^{},u_e^{}) \,\smcirc\,
   \pi_{JG}^{\mathrm{fr}}(u_g^{})^{-1} \,,
\end{equation}
where $T\Phi_E^{}$ denotes the tangent map to~$\Phi_E^{}$
and $\, \pi_{JG}^{\mathrm{fr}}: JG \longrightarrow GL(TM) \,$
is the natural projection of~$JG$ to the linear frame groupoid
$GL(TM)$ of the base manifold~$M$ defined by
\begin{equation} \label{eq:PRJGFRG1}
 \pi_{JG}^{\mathrm{fr}}(u_g^{})~
 =~T_g^{} \tau_G^{} \,\smcirc\, u_g^{} \,,
\end{equation}
whereas $\, \pi_{JG}^{}: JG \longrightarrow G \,$ is the usual
jet target projection.
(Cf.\ equations~(51), (93) and~(94) of Ref.~\cite{CFP}.)
This definition can also be phrased in terms of (bi)sections, as
follows: given any bisection $\beta$ of~$G$ and any section
$\varphi$ of~$E$, concatenate them into a map $(\beta,\varphi)$
from~$M$ to $G \times_M E$ and compose that with the action
$\Phi_E^{}$ of~$G$ on~$E$ to produce a map from~$M$ to~$E$
which, when precomposed with the inverse of the diffeomorphism
$\tau_G^{} \smcirc \beta$ of~$M$ induced by~$\beta$, gives a
new section $\, \Phi_E^{} \,\smcirc\, (\beta,\varphi) \,\smcirc\,
(\tau_G^{} \smcirc \beta)^{-1} \,$ of~$E$, and $\Phi_{JE}^{}$
is then fully characterized by the property that, upon taking the jet
prolongations of all these (bi)sections,
\begin{equation} \label{eq:IAJ1GJ1B3}
 \Phi_{JE}^{} \,\smcirc\, (j\beta,j\varphi) \,\smcirc\,
 (\tau_G^{} \smcirc \beta)^{-1}~
 =~j \bigl( \Phi_E^{} \,\smcirc\, (\beta,\varphi) \,\smcirc\,
   (\tau_G^{} \smcirc \beta)^{-1} \bigr) \,.
\end{equation}
Indeed, for any $y \in M$, putting $\, x = (\tau_G^{} \smcirc
\beta)^{-1}(y) \in M$, we have $\, (\beta(x),\varphi(x)) \in
G \times_M E$, $(j\beta(x),j\varphi(x)) \in JG \times_M JE \,$
and
\[
 \pi_{JG}^{\mathrm{fr}}(j\beta(x))^{-1}~
 =~(T_{\beta(x)} \tau_G^{} \,\smcirc\, T_x^{} \beta)^{-1}~
 =~(T_x^{} (\tau_G^{} \smcirc \beta))^{-1}~
 =~T_y^{} \bigl( (\tau_G^{} \smcirc \beta)^{-1} \bigr) \,,
 \]
 so
\[
\begin{aligned}
 &\bigl( \Phi_{JE}^{} \,\smcirc\, (j\beta,j\varphi) \,\smcirc\,
         (\tau_G^{} \smcirc \beta)^{-1} \bigr)(y)~
  =~\Phi_{JE}^{}(j\beta(x),j\varphi(x)) \\[1mm]
 & \qquad =~T_{(\beta(x),\varphi(x))} \Phi_E^{} \,\smcirc\,
            (T_x^{} \beta,T_x^{} \varphi) \,\smcirc\,
            T_y^{} \bigl( (\tau_G^{} \smcirc \beta)^{-1} \bigr) \\[1mm]
 & \qquad =~T_y^{} \bigl( \Phi_E^{} \,\smcirc\, (\beta,\varphi)
            \,\smcirc\, (\tau_G^{} \smcirc \beta)^{-1} \bigr)~
            =~j \bigl( \Phi_E^{} \,\smcirc\, (\beta,\varphi) \,\smcirc\,
            (\tau_G^{} \smcirc \beta)^{-1} \bigr) (y) \,.
\end{aligned}
\]
Now we have the following statement about compatibility between
these various actions:
\begin{prp}~\label{prp:EQUIV1}
 The difference map of equation~(\ref{eq:DIFMAP1}) is equivariant, i.e.,
 the diagram
 \begin{equation}
  \begin{array}{c}
   \xymatrix{
    ~~~JG \times_M (JE \times_E JE)~~~~
    \ar[r] \ar[d]_{(\pi_{JG}^{\mathrm{fr}} \times \pi_{JG}^{}\,,\,-)\,} &
    ~JE \times_E JE \ar[d]^{\,-} \\
    (GL(TM) \times_M G) \times_M \vec{J} E~ \ar[r] &
    ~~~~~\vec{J} E~~~~
   }
  \end{array}
 \end{equation}
 commutes.
\end{prp}

\begin{proof}
 Given $g \in G$ with $\sigma_G^{}(g) = x$ and $\tau_G^{}(g) = y$,
 $e \in E$ with $\pi(e) = x$, $u_g^{} \in J_g^{} G$ and $\, u_e^1,
 u_e^2 \in J_e^{} E \subset L(T_x^{} M,T_e^{} E)$, we want to
 prove that
 \[
  u_g^{} \!\cdot u_e^2 - u_g^{} \!\cdot u_e^1~
  =~(\pi_{JG}^{\mathrm{fr}}(u_g^{}) , g) \cdot (u_e^2 - u_e^1) \,.
 \]
 Fixing some tangent vector $v \in T_x M$, choose a vertical curve
 $e(t)$ in $E$ ($\pi(e(t)) = x$) such that
 \[
  e(t) \big|_{t=0}~=~e~~~,~~~
  \frac{d}{dt} \, e(t) \Big|_{t=0}~=~(u_e^2 - u_e^1)(v) \,.
 \]
 Then
 \[
 \begin{aligned}
  &T_{(g,e)} \Phi_E^{} \bigl( u_g^{}(v) , u_e^2(v) \bigr) \, - \,
   T_{(g,e)} \Phi_E^{} \bigl( u_g^{}(v) , u_e^1(v) \bigr)~
   =~T_{(g,e)} \Phi_E^{} \, \bigl( 0 , (u_e^2 - u_e^1)(v) \bigr)
  \\[1ex]
  & \qquad =~\frac{d}{dt} \, \Phi_E^{}(g,e(t)) \Big|_{t=0}~
   =~\frac{d}{dt} \, L_g^{}(e(t)) \Big|_{t=0}~
   =~T_e^{} L_g^{} \, \bigl( (u_e^2 - u_e^1)(v) \bigr) \,,
 \end{aligned}
 \]
 or using that $v$ was arbitrary,
 \[
  T_{(g,e)} \Phi_E^{} \,\smcirc\, \bigl( u_g^{} , u_e^2 \bigr) \, - \,
  T_{(g,e)} \Phi_E^{} \,\smcirc\, \bigl( u_g^{} , u_e^1 \bigr)~
  =~T_{(g,e)} \Phi_E^{} \,\smcirc\, \bigl( 0 , u_e^2 - u_e^1 \bigr)~
  =~T_e^{} L_g^{} \,\smcirc\, (u_e^2 - u_e^1) \,.
 \]
 Precomposing with $\pi_{JG}^{\mathrm{fr}}(u_g^{})^{-1}$ proves
 the claim.
\qed
\end{proof}

To deal with the second part, we begin by iterating the procedure of
applying the jet functor to obtain an induced action
\begin{equation} \label{eq:IAIJGIJB1}
 \begin{array}{cccc}
  \Phi_{J(JE)}:
  & J(JG) \times_M J(JE) & \longrightarrow &          J(JE)
  \\[1mm]
  & (u'_{u_g},u'_{u_e})  &   \longmapsto   & u'_{u_g} \!\cdot\, u'_{u_e}
 \end{array}
\end{equation}
of the iterated jet groupoid~$J(JG)$ of~$G$ on the iterated jet bundle~%
$J(JE)$ of~$E$, defined by
\begin{equation} \label{eq:IAIJGIJB2}
 u'_{u_g} \!\cdot\, u'_{u_e}~
 =~T_{(u_g,u_e)} \Phi_{JE}^{} \,\smcirc\, (u'_{u_g},u'_{u_e})
   \,\smcirc\, \pi_{J(JG)}^{\mathrm{fr}}(u'_{u_g})^{-1} \,,
\end{equation}
with the same notation as before; in particular, the definition can again
be phrased in terms of (bi)sections.
Namely, given any bisection $\tilde{\beta}$ of~$JG$ and any section
$\tilde{\varphi}$ of~$JE$ which (by composition with $\pi_{JG}^{}$)
project to a bisection $\beta$ of~$G$ and to a section~$\varphi$
of~$E$, respectively, so that $\, \tau_{JG}^{} \smcirc \tilde{\beta}
= \tau_G^{} \smcirc \beta$, we have, just as in equation~(\ref%
{eq:IAJ1GJ1B3}) above,
\begin{equation} \label{eq:IAIJGIJB3}
 \Phi_{J(JE)}^{} \,\smcirc\, (j\tilde{\beta},j\tilde{\varphi})
 \,\smcirc\, (\tau_G^{} \smcirc \beta)^{-1}~
 =~j \bigl( \Phi_{JE}^{} \,\smcirc\, (\tilde{\beta},\tilde{\varphi})
            \,\smcirc\, (\tau_G^{} \smcirc \beta)^{-1} \bigr) \,.
\end{equation}
This iterated action admits restrictions to several subgroupoids and
subbundles, among which the following will become important to us
at some point or another: the natural induced actions
\begin{equation} \label{eq:IASJGSJB1}
 \begin{array}{cccc}
  \Phi_{\bar{J}^{\>\!2} E}:
  & \bar{J}^{\>\!2} G \times_M \bar{J}^{\>\!2} E
  & \longrightarrow & \bar{J}^{\>\!2} E \\[1mm]
  & (u'_{u_g},u'_{u_e})  &   \longmapsto   & u'_{u_g} \!\cdot\, u'_{u_e} 
 \end{array}
\end{equation}
of the semiholonomous second order jet groupoid~$\bar{J}^{\>\!2} G$
of~$G$ and
\begin{equation} \label{eq:IAJ2GSJB1}
 \begin{array}{cccc}
  \Phi_{\bar{J}^{\>\!2} E}:
  & J^{\>\!2} G \times_M \bar{J}^{\>\!2} E
  & \longrightarrow & \bar{J}^{\>\!2} E \\[1mm]
  & (u'_{u_g},u'_{u_e})  &   \longmapsto   & u'_{u_g} \!\cdot\, u'_{u_e}
 \end{array}
\end{equation}
of the second order jet groupoid~$J^{\>\!2} G$ of~$G$ on the
semiholonomous second order jet bundle $\bar{J}^{\>\!2} E$ of~$E$,
as well as the action
\begin{equation} \label{eq:IAJ2GJ2B1}
 \begin{array}{cccc}
  \Phi_{J^{\>\!2} E}:
  & J^{\>\!2} G \times_M J^{\>\!2} E
  & \longrightarrow & J^{\>\!2} E \\[1mm]
  & (u'_{u_g},u'_{u_e})  &   \longmapsto   & u'_{u_g} \!\cdot\, u'_{u_e}
 \end{array}
\end{equation}
of the second order jet groupoid~$J^{\>\!2} G$ of~$G$ on the
second order jet bundle $J^{\>\!2} E$ of~$E$, all defined by the
same formula:
\begin{equation} \label{eq:IASJGSJB2}
 u'_{u_g} \!\cdot\, u'_{u_e}~
 =~T_{(u_g,u_e)} \Phi_{JE}^{} \,\smcirc\, (u'_{u_g},u'_{u_e})
   \,\smcirc\, \pi_{JG}^{\mathrm{fr}}(u_g^{})^{-1} \,.
\end{equation}
Here, the simplification in the last term on the rhs of equation~%
(\ref{eq:IASJGSJB2}), as compared to that of equation~%
(\ref{eq:IAIJGIJB2}), stems from the fact that when
$\, u'_{u_g} \in \bar{J}_{u_g}^{\>\!2} G$, i.e.,
$T_{u_g}^{} \pi_{JG}^{} \,\smcirc\, u'_{u_g} = u_g^{}$,
then since $\, \tau_{JG}^{} = \tau_G^{} \,\smcirc\, \pi_{JG}^{}$,
we~get
\[
 \pi_{J(JG)}^{\mathrm{fr}}(u'_{u_g})~
 =~T_ {u_g}^{} \tau_{JG}^{} \,\smcirc\, u'_{u_g}~
 =~T_ g^{} \tau_G^{} \,\smcirc\,
   T_ {u_g}^{} \pi_{JG}^{} \,\smcirc\, u'_{u_g}~
 =~T_g^{} \tau_G^{} \,\smcirc\, u_g^{}~
 =~\pi_{JG}^{\mathrm{fr}}(u_g^{}) \,.
\]
Moreover, if $u'_{u_g}$ and $u'_{u_e}$ are both semiholonomous,
then so is $u'_{u_g} \!\cdot\, u'_{u_e}$, i.e., we have
\[
 u'_{u_g} \in \bar{J}_{u_g}^{\>\!2} G \,,\,
 u'_{u_e} \in \bar{J}_{u_e}^{\>\!2} E~\Longrightarrow~
 u'_{u_g} \!\cdot\, u'_{u_e} \in \bar{J}_{u_g \cdot u_e}^{\>\!2} E \,,
\]
since in this case, $T_{u_g}^{} \pi_{JG}^{} \,\smcirc\, u'_{u_g} = u_g^{} \,$
and $\, T_{u_e}^{} \pi_{JE}^{} \,\smcirc\, u'_{u_e} = u_e^{}$, and using the
equality $\, \pi_{JE}^{} \,\smcirc\, \Phi_{JE}^{} = \Phi_E^{} \,\smcirc\,
(\pi_{JG}^{} \times_M^{} \pi_{JE}^{})$, we get
\[
\begin{aligned}
 &T_{u_g \cdot\, u_e}^{} \pi_{JE}^{} \,\smcirc\,
  (u'_{u_g} \!\cdot\, u'_{u_e})~
  =~T_{u_g \cdot\, u_e}^{} \pi_{JE}^{} \,\smcirc\,
    T_{(u_g,u_e)} \Phi_{JE}^{} \,\smcirc\, (u'_{u_g},u'_{u_e})
    \,\smcirc\, \pi_{JG}^{\mathrm{fr}}(u_g^{})^{-1} \\[1mm]
 & \qquad =~T_{(g,e)} \Phi_E^{} \,\smcirc\,
            \bigl( T_{u_g}^{} \pi_{JG}^{} \,\smcirc\, u'_{u_g} \,,\,
                   T_{u_e}^{} \pi_{JE}^{} \,\smcirc\, u'_{u_e} \bigr)
            \,\smcirc\, \pi_{JG}^{\mathrm{fr}}(u_g^{})^{-1} \\[1mm]
 & \qquad =~T_{(g,e)} \Phi_E^{} \,\smcirc\, (u_g^{},u_e^{})
            \,\smcirc\, \pi_{JG}^{\mathrm{fr}}(u_g^{})^{-1} \\[1mm]
 & \qquad =~u_g^{} \cdot u_e^{} \,.
\end{aligned}
\]
Similarly, it is clear that if $u'_{u_g}$ and $u'_{u_e}$ are both holonomous,
then so is $u'_{u_g} \!\cdot\, u'_{u_e}$, i.e., we have
\[
 u'_{u_g} \in J_{u_g}^{\>\!2} G \,,\,
 u'_{u_e} \in J_{u_e}^{\>\!2} E~\Longrightarrow~
 u'_{u_g} \!\cdot\, u'_{u_e} \in J_{u_g \cdot u_e}^{\>\!2} E \,,
\]
since in this case there will exist a local bisection $\beta$ of~$G$ and a
local section $\varphi$ of~$E$, both defined in some open neighborhood
$U$ of~$x$, satisfying $g = \beta(x)$, $e = \varphi(x)$, $u_g^{} =
j\beta(x) = T_x^{} \beta$, $u_e^{} = j\varphi(x) = T_x^{} \varphi$,
$u'_{u_g} = j(j\beta)(x) = T_x^{} (j\beta)$, $u'_{u_e} = j(j\varphi)(x)
= T_x^{} (j\varphi) \,$ and hence, putting \linebreak $y = (\tau_G^{}
\smcirc \beta)(x) \,$ and using equation~(\ref{eq:IASJGSJB2}), 
equation~(\ref{eq:IAIJGIJB3}) with $\, \tilde{\beta} = j\beta$,
$\tilde{\varphi} = j\varphi \,$ and equation~(\ref{eq:IAJ1GJ1B3}),
\[
\begin{aligned}
 u'_{u_g} \!\cdot\, u'_{u_e}~
 &=~\Phi_{J(JE)}^{}(j(j\beta)(x),j(j\varphi)(x)) \\[1mm]
 &=~\bigl( \Phi_{J(JE)}^{} \,\smcirc\, (j(j\beta),j(j\varphi))
           \,\smcirc\, (\tau_G^{} \smcirc \beta)^{-1} \bigr) (y) \\[1mm]
 &=~j \bigl( \Phi_{JE}^{} \,\smcirc\, (j\beta,j\varphi)
                     \,\smcirc\, (\tau_G^{} \smcirc \beta)^{-1}
                     \bigr) (y) \\[1mm]
 &= j \bigl( j \bigl( \Phi_E^{} \,\smcirc\, (\beta,\varphi)
                              \,\smcirc\, (\tau_G^{} \smcirc \beta)^{-1}
                              \bigr) \bigr) (y) \,.
\end{aligned}
\]
Finally, observe that, just like the (first order) jet groupoid~$JG$
of~$G$, its iterated jet groupoid $J(JG)$ and, by restriction, its
semiholonomous second order jet groupoid~$\bar{J}^{\>\!2} G$
and second order jet groupoid~$J^{\>\!2} G$ all admit natural
projections both to $GL(TM)$ and to~$G$, which are just given
by composition of those for~$JG$ with the natural projection
$\, \pi_{J(JG)}^{}: J(JG) \longrightarrow JG \,$ and its respective
restrictions $\, \pi_{\bar{J}^{\>\!2} G}^{}: \bar{J}^{\>\!2} G
\longrightarrow JG \,$ and $\, \pi_{J^{\>\!2} G}^{}: J^{\>\!2} G
\longrightarrow JG$:
\[
 \begin{array}{c}
  \pi_{J(JG)}^{\mathrm{fr}}~
  =~\pi_{JG}^{\mathrm{fr}} \,\smcirc\, \pi_{J(JG)}^{}:
  J(JG) \longrightarrow GL(TM)~~~,~~~
  \pi_{J(JG),G}^{}~
  =~\pi_{JG}^{} \,\smcirc\, \pi_{J(JG)}^{}:
  J(JG) \longrightarrow G
  \\[1ex]
  \pi_{\bar{J}^{\>\!2} G}^{\mathrm{fr}}~
  =~\pi_{JG}^{\mathrm{fr}} \,\smcirc\, \pi_{\bar{J}^{\>\!2} G}^{}:
  \bar{J}^{\>\!2} G \longrightarrow GL(TM)~~~,~~~
  \pi_{\bar{J}^{\>\!2} G,G}^{}~
  =~\pi_{JG}^{} \,\smcirc\, \pi_{\bar{J}^{\>\!2} G}^{}:
  \bar{J}^{\>\!2} G \longrightarrow G
  \\[1ex]
  \pi_{J^{\>\!2} G}^{\mathrm{fr}}~
  =~\pi_{JG}^{\mathrm{fr}} \,\smcirc\, \pi_{J^{\>\!2} G}^{}:
  J^{\>\!2} G \longrightarrow GL(TM)~~~,~~~
  \pi_{J^{\>\!2} G,G}^{}~
  =~\pi_{JG}^{} \,\smcirc\, \pi_{J^{\>\!2} G}^{}:
  J^{\>\!2} G \longrightarrow G \\[1mm]
 \end{array}
\]
With this notation, we can now formulate the following statement about
compatibility between these various actions:
\begin{prp}~\label{prp:EQUIV2}
 The difference maps of equations~(\ref{eq:DIFMAP2}) and~(\ref{eq:DIFMAP3})
 are equivariant, i.e., the diagrams
 \begin{equation}
  \begin{array}{c}
   \xymatrix{
    \qquad\qquad
    \bar{J}^{\>\!2} G \times_M
    (\bar{J}^{\>\!2} E \times_{JE} \bar{J}^{\>\!2} E)
    ~\qquad\qquad
    \ar[r] \ar[d]_{(\pi_{\bar{J}^{\>\!2} G}^{\mathrm{fr}}
                    \times \pi_{\bar{J}^{\>\!2} G,G}^{}\,,\,-)\,} &
    ~~~~\bar{J}^{\>\!2} E \times_{JE} \bar{J}^{\>\!2} E~~~
    \ar[d]^{\,-} \\
    (GL(TM) \times_M G) \times_M^{}
    \Bigl( \pi^* \bigl( \bigotimes\nolimits^{\!2} T^\ast M \bigr)
           \otimes VE \Bigr)~ \ar[r] &
    ~\pi^* \bigl( \bigotimes\nolimits^{\!2} T^\ast M \bigr) \otimes VE
   }
  \end{array}
 \end{equation}
 and
 \begin{equation}
  \begin{array}{c}
   \xymatrix{
    \qquad\qquad
    J^{\>\!2} G \times_M (J^{\>\!2} E \times_{JE} J^{\>\!2} E)
    ~\qquad\qquad
    \ar[r] \ar[d]_{(\pi_{J^{\>\!2} G}^{\mathrm{fr}}
                    \times \pi_{J^{\>\!2} G,G}^{}\,,\,-)\,} &
    ~~~~J^{\>\!2} E \times_{JE} J^{\>\!2} E~~~
    \ar[d]^{\,-} \\
    (GL(TM) \times_M G) \times_M^{}
    \Bigl( \pi^* \bigl( \bvee^{\;\!2\,} T^\ast M \bigr)
           \otimes VE \Bigr)~ \ar[r] &
    ~\pi^* \bigl( \bvee^{\;\!2\,} T^\ast M \bigr) \otimes VE
   }
  \end{array}
 \end{equation}
 commute. Similarly, the alternator or antisymmetrizer map of equation~%
 (\ref{eq:ASPSSOJ}) is also equivariant, i.e., the diagram
 \begin{equation}
  \begin{array}{c}
   \xymatrix{
    \qquad\qquad\qquad~
    J^{\>\!2} G \times_M \bar{J}^{\>\!2} E
    ~~\qquad\qquad\qquad
    \ar[r] \ar[d]_{(\pi_{J^{\>\!2} G}^{\mathrm{fr}}
                    \times \pi_{J^{\>\!2} G,G}^{}\,,\,\mathrm{Alt})\,} &
    \quad\qquad~ \bar{J}^{\>\!2} E \qquad\quad
    \ar[d]^-{\,\mathrm{Alt}} \\
    (GL(TM) \times_M G) \times_M^{}
    \Bigl( \pi^* \bigl( \bwedge^{\!2\,} T^\ast M \bigr)
           \otimes VE \Bigr)~ \ar[r] &
    ~\pi^* \bigl( \bwedge^{\!2\,} T^\ast M \bigr) \otimes VE
   }
  \end{array}
 \end{equation}
 commutes.
\end{prp}

\begin{proof}
 First of all, the statements about commutativity of the last two
 diagrams  are trivial consequences of that about commutativity of
 the first, together with the fact that the decomposition of rank~$2$
 tensors into their symmetric and antisymmetric parts is obviously
 invariant under the action of~$GL(TM) \times_M G$.
 To deal with the first diagram, we shall find it convenient to keep track
 of the identifications made in the definition of the difference map in
 equation~(\ref{eq:DIFMAP2}) by momentarily (i.e., just for the
 remainder of this proof) denoting that difference map by $\delta$.
 Thus given $g \in G$ with $\sigma_G^{}(g) = x$ and $\tau_G^{}(g) = y$,
 $e \in E$ with $\pi(e) = x$, $u_g^{} \in J_g^{} G$, $u_e^{} \in J_e^{} E$,
 $u'_{u_g} \in \bar{J}_{u_g}^{\>\!2} G \,$ and $\, u_{u_e}^{\prime\,1},
 u_{u_e}^{\prime\,2} \in \bar{J}_{u_e}^{\>\!2} E \subset J_{u_e}^{}(JE)
 \subset L(T_x^{} M,T_{u_e}^{}(JE))$, we want to show that
 \[
  \delta( u'_{u_g} \!\cdot\, u_{u_e}^{\prime\,2} ,
          u'_{u_g} \!\cdot\, u_{u_e}^{\prime\,1})~
  =~(\pi_{JG}^{\mathrm{fr}}(u_g^{}) , g) \cdot
    \delta(u_{u_e}^{\prime\,2} , u_{u_e}^{\prime\,1}) \,.
 \]
 Note that $\, \delta(u_{u_e}^{\prime\,2},u_{u_e}^{\prime\,1})
 \in L^2(T_x^{} M,V_e^{} E) \,$ can be defined explicitly by stating
 that, for any tangent vector $v \in T_x M$, the standard difference
 $\, u_{u_e}^{\prime\,2} - u_{u_e}^{\prime\,1}$, when evaluated
 on~$v$, gives a tangent vector in $T_{u_e}(JE)$ which, being vertical
 with respect to the jet target projection $\pi_{JE}^{}$, can be
 realized as that of a straight line in~$J_e^{} E$ through~$u_e^{}$,
 whose direction is $\, \delta(u_{u_e}^{\prime\,2},u_{u_e}^{\prime\,1})
 (v,.) \in L(T_x^{} M,V_e^{} E)$:
 \[
  (u_{u_e}^{\prime\,2} - u_{u_e}^{\prime\,1})(v)~
  =~\frac{d}{dt} \bigl( u_e^{} \, + \,
    t \, \delta(u_{u_e}^{\prime\,2},u_{u_e}^{\prime\,1})(v,.) \bigr)
    \big|_{t=0} \,.
 \]
 Similarly, $\delta(u'_{u_g} \!\cdot\, u_{u_e}^{\prime\,2} , u'_{u_g}
 \!\cdot\,u_{u_e}^{\prime\,1}) \in L^2(T_y^{} M,V_{g \cdot e}^{} E) \,$
 can be defined explicitly by stating that, for any tangent vector $w \in
 T_y M$, the standard difference $\, u'_{u_g} \!\cdot\, u_{u_e}^{\prime\,2}
 - u'_{u_g} \!\cdot\, u_{u_e}^{\prime\,1}$, when evaluated on~$w$, gives
 a tangent vector in $T_{u_g \cdot u_e}(JE)$ which, being vertical with respect
 to the jet target projection $\pi_{JE}^{}$, can be realized as that
 of a straight line in $J_{g \cdot e}^{} E$ through~$u_g^{}
 \cdot u_e^{}$, whose direction is $\, \delta(u'_{u_g} \!\cdot\,
 u_{u_e}^{\prime\,2}, u'_{u_g} \!\cdot\, u_{u_e}^{\prime\,1})
 (w,.) \in L(T_y^{} M,V_{g \cdot e}^{} E)$:
 \[
  (u'_{u_g} \!\cdot\, u_{u_e}^{\prime\,2} \, - \,
   u'_{u_g} \!\cdot\, u_{u_e}^{\prime\,1})(w)~
  =~\frac{d}{dt} \bigl( u_g^{} \cdot u_e^{} \, + \,
    t \, \delta(u'_{u_g} \!\cdot\, u_{u_e}^{\prime\,2},
                u'_{u_g} \!\cdot\, u_{u_e}^{\prime\,1})(w,.) \bigr)
    \big|_{t=0} \,.
 \]
 On the other hand, putting $\, v = \pi_{JG}^{\mathrm{fr}}
 (u_g)^{-1}(w)$, we have
 \[
 \begin{aligned}
  &(u'_{u_g} \!\cdot\, u_{u_e}^{\prime\,2} \, - \,
   u'_{u_g} \!\cdot\, u_{u_e}^{\prime\,1})(w) \\[1ex]
  & \qquad =~T_{(u_g,u_e)} \Phi_{JE}^{}
             \bigl( u'_{u_g}(v) , u_{u_e}^{\prime\,2}(v) \bigr) \, - \,
             T_{(u_g,u_e)} \Phi_{JE}^{}
             \bigl( u'_{u_g}(v) , u_{u_e}^{\prime\,1}(v) \bigr) \\[1.5ex]
  & \qquad =~T_{(u_g,u_e)} \Phi_{JE}^{} \bigl( 0 ,
          (u_{u_e}^{\prime\,2} - u_{u_e}^{\prime\,1})(v) \bigr) \\[0.5ex]
  & \qquad =~\frac{d}{dt} \, \Phi_{JE}^{} \bigl( u_g^{} , u_e^{} \, + \,
             t \, \delta(u_{u_e}^{\prime\,2},u_{u_e}^{\prime\,1})(v,.)
             \bigr) \Big|_{t=0} \\[1ex]
  & \qquad =~\bigl( (\pi_{JG}^{\mathrm{fr}}(u_g^{}) , g) \cdot
                    \delta(u_{u_e}^{\prime\,2},u_{u_e}^{\prime\,1})
                    \bigr) (w,.) \,,
 \end{aligned}
 \]
 where in the last step we have used the fact that, as shown in
 Ref.~\cite{CFP}, the action $\Phi_{JE}^{}$ is affine along the
 fibers of~$JE$ over~$E$, together with Proposition~\ref{prp:EQUIV1}.
 \qed
\end{proof}

Returning to the formalization of the minimal coupling prescription
and the curvature map, we want to emphasize that the context
outlined above is a little bit too broad to fit into the theoretical
setting of field theory, since general connections in general fiber
bundles are \emph{not} fields! \linebreak
This is so because they are not sections of bundles over space-time but
rather sections of bundles over some ``extended space-time'' which
is itself the total space of some fiber bundle over ordinary space-time.
As such, when expressed in local coordinates and local trivializations,
such sections correspond to multiplets of functions which, apart from
being functions on space-time, depend on extra ``vertical'' variables,
namely, the local coordinates along the fibers of this bundle, and in
 the absence of stringent restrictions on that dependence will produce
\emph{infinite} multiplets of fields when expanded in an appropriate
basis.
This situation is familiar from ``Kaluza-Klein'' type theories, which
have been proposed long ago as models for unifying gravity with the
other fundamental interactions and where the extended space-time is
assumed to be the total space of some principal bundle over ordinary
space-time, so that one can use the representation theory of the
underlying structure group to control and restrict the dependence
of functions on the extra vertical variables.%
\footnote{The simplest such model and one of the most interesting
attempts to unify gravity with electromagnetism uses an extended
space-time which is the total space of a principal $U(1)$-bundle
over ordinary space-time, so the extra vertical variables reduce
to a single phase~$\theta$, the representations of the structure
group are given by its characters $\, \theta \longmapsto \exp
(ik\theta)$, $k \in \mathbb{Z}$, and the expansion of functions on
extended space-time is just a Fourier expansion with coefficients
that are functions on ordinary space-time: still an infinite multiplet
of fields.} 
The main problem with these models is that the aforementioned
stringent restrictions, needed to weed out the large number of
(often unwanted) extra fields, are usually quite artificial and
imposed more or less ``ad hoc'', without any convincing argument
as to how they should arise from the dynamics of a fundamental
theory in higher dimensions.

Here, these remarks serve merely as a guide to what should be
done and what not: we shall completely avoid all these problems
by working not with general connections but only with connections
that do have a natural interpretation as fields in physics: these
are connections whose behavior along the fibers is fixed by some
condition, such as linear connections in vector bundles or affine
connections in affine bundles, where the connection coefficients
are required to be linear or affine functions along the fibers,
respectively, or more generally, principal connections, which
are required to be equivariant under the action of the structure
group on the fibers of the principal bundle and are therefore
completely fixed along the entire fiber once they are known
at a single point in that fiber.

Thus from this point onward and throughout the rest of the paper,
we shall assume that $E$ is not just a general fiber bundle but
rather a fiber bundle with structure group, which is a Lie group~%
$G_0^{}$, with Lie algebra~$\mathfrak{g}_0^{}$, say, so there
is a principal $G_0^{}$-bundle $P$ to which $E$ is associated (this,
by the way, includes the case where $E$ is~$P$ itself), and any
connection in~$E$ to be considered is associated to a principal
connection in~$P$.
As a result, we have to adapt our formalism to this situation, and
of course the Lie groupoid~$G$ that appears above, as well as in
Ref.~\cite{CFP}, but has so far been left unspecified, will now be
the gauge groupoid of~$P$.

\section{Gauge groupoids, jet groupoids and induced actions}

In order to implement the program outlined in the last paragraph
of the previous section, we shall first introduce the gauge groupoid
of a principal bundle and some of its actions (more specifically, on
the principal bundle itself and on any of its associated bundles, as
well as on the respective vertical bundles) and then investigate how
some of these lift when taking first and second order jet prolongations.

\subsection{The gauge groupoid and its actions}

To begin with, let us recall the definition of the gauge groupoid of
a principal bundle~\cite{Mac}:
\begin{prp} \label{prp:GGRPD}~
 Given a principal bundle $P$ over a manifold $M$ with structure
 group $G_0^{}$, whose bundle projection will be denoted by $\, \rho:
 P \longrightarrow M$, let
 \[
  G = (P \times P)/G_0^{}
 \]
 denote the orbit space of the cartesian product of~$P$ with itself
 under the diagonal action of~$G_0^{}$ (we shall write its elements as
 classes $[p_2^{},p_1^{}]$ of pairs $(p_2^{},p_1^{})$ in $P \times P$,
 where $\, [p_2^{} \cdot g_0^{},p_1^{} \cdot g_0^{}] = [p_2^{},p_1^{}]$).
 Then $G$ is a Lie groupoid over~$M$, called the\/ \textbf{gauge groupoid}
 of~$P$, with source projection $\, \sigma_G^{}: G \longrightarrow M$,
 target projection $\, \tau_G^{}: G \longrightarrow M$, multiplication
 map $\, \mu_G^{}: G \times_M G \longrightarrow G$, unit map $\, 1_G^{}:
 M \longrightarrow G \,$ and inversion $\iota_G^{}: G \longrightarrow G \,$
 defined as follows:
 \begin{itemize}
  \item for $[p_2^{},p_1^{}] \in G$,
        \[
         \sigma_G^{}([p_2^{},p_1^{}]) = \rho(p_1^{})~~,~~
         \tau_G^{}([p_2^{},p_1^{}]) = \rho(p_2^{}) \,;
        \]
  \item for $[p_2^{},p_1^{}],[p_3^{},p_2^{}] \in G$,
        \[
         [p_3^{},p_2^{}][p_2^{},p_1^{}] \equiv
         \mu_G^{}([p_3^{},p_2^{}],[p_2^{},p_1^{}]) = [p_3^{},p_1^{}] \,;
        \]
  \item for $x \in M$,
        \[
         (1_G^{})_x^{} = [p,p] \,,
        \]
        where $p$ is any element of $\rho^{-1}(x)$;
  \item for $[p_2^{},p_1^{}] \in G$,
        \[
         [p_2^{},p_1^{}]^{-1} \equiv
         \iota_G^{}([p_2^{},p_1^{}]) = [p_1^{},p_2^{}] \,.
        \]
        \vspace{-2ex}
 \end{itemize}
\end{prp}
Observe that the gauge group bundle associated with $P$ employed in
Ref.~\cite{FS}, also known as the adjoint bundle $\, \mathrm{Ad} P
= P \times_{G_0} G_0^{}$ (where $G_0^{}$ acts on itself by
conjugation), is (up to a canonical isomorphism) just the isotropy
subgroupoid of~$G$, that is,
\begin{equation}
 P \times_{G_0} G_0^{} \,\cong\, G_{\mathrm{iso}} \,.
\end{equation}
This isomorphism can be constructed explicitly by noting that the map
\[
 \begin{array}{ccc}
  P \times G_0^{} & \longrightarrow &   P \times P    \\[1mm]
     (p,g_0^{})   &   \longmapsto   & (p,p \cdot g_0)
 \end{array}
\]
is equivariant under the right action of~$G_0^{}$ on both sides
(since it takes $(p \cdot g_0' \,, (g_0')^{-1} g_0^{} g_0')$ to 
$(p \cdot g_0' \,, p \cdot g_0^{} g_0')$) and hence factors
to the respective quotients to yield a map
\[
 \begin{array}{ccc}
  P \times_{G_0} G_0^{} & \longrightarrow & (P \times P)/G_0 \\[1mm]
        [p,g_0^{}]      &   \longmapsto   & [p,p \cdot g_0]
 \end{array}
\]
which is the desired isomorphism onto its image
\begin{equation}
 G_{\mathrm{iso}} \,
 = \, \{ [p_2^{},p_1^{}] \in G \,|\,
         \tau_G^{}([p_2^{},p_1^{}]) = \sigma_G^{}([p_2^{},p_1^{}]) \} \,
 = \, \{ [p_2^{},p_1^{}] \in G \,|\, \rho(p_2^{}) = \rho(p_1^{}) \} \,.
\end{equation}
Moreover, it is well known that the group of bisections of the gauge
groupoid~$\, G = (P \times P)/G_0^{}$ \linebreak is isomorphic to
the group of automorphisms of~$P$,
\begin{equation}
 \mathrm{Bis}(G) \cong \mathrm{Aut}(P) \,,
\end{equation}
while the group of sections of the gauge group bundle $\, G_{\mathrm{iso}}
\cong P \times_{G_0} G_0^{} \,$ is isomorphic to the group of strict
automorphisms of $P$,
\begin{equation}
 \Gamma(G_{\mathrm{iso}}) \cong \mathrm{Aut}_s(P) \,.
\end{equation}

Next, let us specify how the gauge groupoid of a principal bundle acts
naturally on the principal bundle itself and on any of its associated
bundles.
To this end, some authors find it convenient to introduce the
``difference map'' for~$P$, which is the smooth map
\[
 \delta_P: P \times_M P \longrightarrow G_0^{}
\]
defined implicitly by the condition that given any two points $p$ and $p'$
in the same fiber of~$P$, $\delta_P(p,p')$ is the unique element of~$G_0^{}$
that transforms $p$ into $p'$:
\[
 p \cdot \delta_P(p,p')~=~p' \,.  
\]
Note that, obviously, $\delta_P(p,p) = 1 \,$ and
\[
 \delta_P^{}(p \cdot g_0^{},p' \cdot g_0^{})~
 =~g_0^{-1} \, \delta_P^{}(p,p') \, g_0^{} \,.
\]
Here, we use this map to write down a natural action
\[
 \begin{array}{cccc}
  \Phi_P^{}:
  &      G \times_M P       & \longrightarrow &              P
  \\[1mm]
  & ([p_2^{},p_1^{}],p) &   \longmapsto   & [p_2^{},p_1^{}] \cdot p 
 \end{array}
\]
of the gauge groupoid $\, G = (P \times P)/G_0^{} \,$ on the principal
bundle $P$ itself, defined as follows: given $\, [p_2^{},p_1^{}] \in G \,$
and $p \in P$ such that $\, \rho(p_1^{}) = \sigma_G^{}([p_2^{},p_1^{}])
= \rho(p)$, put
\[
 [p_2^{},p_1^{}] \cdot p~=~p_2^{} \cdot \delta_P(p_1^{},p) \,.
\]
Note, however, that we can always adapt the second component in the
pair $(p_2,p_1)$ representing the class $[p_2,p_1]$ to be equal to~$p$,
which allows us to rewrite the previous two equations in the simplified form
\begin{equation} \label{eq:ACTGGPB1}
 \begin{array}{cccc}
  \Phi_P^{}:
  & G \times_M P & \longrightarrow &       P        \\[1mm]
  &  ([p',p],p)  &   \longmapsto   & [p',p] \cdot p 
 \end{array}
\end{equation}
where
\begin{equation} \label{eq:ACTGGPB2}
 [p',p] \cdot p~=~p' \,.
\end{equation}
In the sequel, when defining other actions of the gauge groupoid,
we shall already perform this kind of simplification right from the
start and without further notice, thus dispensing the need to deal
with the difference map $\delta_P$ altogether.
Of course, as the total space of a principal bundle, $P$ also carries
a right action of the structure group~$G_0^{}$, and remarkably,
these two actions commute, 
\begin{equation} \label{eq:ACTGGPB5}
 [p',p] \cdot (p \cdot g_0^{})~=~([p',p] \cdot p) \cdot g_0^{} \,,
\end{equation}
because both sides are equal to $\, [p' \cdot g_0^{},p \cdot g_0^{}]
\cdot (p \cdot g_0^{}) = p' \cdot g_0^{}$.
Thus using the natural projection of~$G$ to the pair groupoid
$M \times M$ of the base manifold~$M$, we get a commutative
diagram:
\begin{equation} \label{eq:ACTGGPB6}
 \begin{array}{c}
  \xymatrix{
   \qquad~ G \times_M P ~\qquad \ar[r] \ar[d] & ~P~ \ar[d] \\
   ~(M \times M) \times_M M~ \ar[r] & ~M~
  }
 \end{array}
\end{equation}
This procedure can be generalized as follows.
First, given any manifold~$Q$, we can introduce a natural action
\begin{equation} \label{eq:ACTGGPR1}
 \begin{array}{cccc}
  \Phi_{P \times Q}^{}:
  & G \times_M (P \times Q) & \longrightarrow &     P \times Q
  \\[1mm]
  &     ([p',p],(p,q))      &   \longmapsto   & [p',p] \cdot (p,q) 
 \end{array}
\end{equation}
of the gauge groupoid $\, G = (P \times P)/G_0^{} \,$ on the product
manifold $P \times Q$ (as a fiber bundle over~$M$), defined by letting
$G$ act as above on the first factor and trivially on the second factor,
\begin{equation} \label{eq:ACTGGPR2}
 [p',p] \cdot (p,q)~=~(p',q) \,.
\end{equation}
Now suppose we are also given a left action
\begin{equation} \label{eq:ACTTF1}
 \begin{array}{ccc}
  G_0^{} \times Q & \longrightarrow &       Q        \\[1mm]
     (g_0^{},q)   &   \longmapsto   & g_0^{} \cdot q
 \end{array}
\end{equation}
of $G_0^{}$ on the manifold~$Q$, which according to the standard
definition of the total space of an associated bundle is extended to
a ``diagonal'' right action
\begin{equation} \label{eq:ACTGGTB1}
 \begin{array}{ccc}
  G_0^{} \times (P \times Q) & \longrightarrow
  & P \times Q \\[1mm]
        (g_0^{},(p,q))       &   \longmapsto   
  & (p \cdot g_0^{},g_0^{-1} \cdot q)
 \end{array}
\end{equation}
of $G_0^{}$ on the product manifold~$P \times Q$, and once again,
these two actions commute,
\begin{equation} \label{eq:ACTGGTB2}
 [p',p] \cdot ((p,q) \cdot g_0^{})~=~([p',p] \cdot (p,q)) \cdot g_0^{} \,,
\end{equation}
because both sides are equal to $\, [p' \cdot g_0^{},p \cdot g_0^{}]
\cdot (p \cdot g_0^{},g_0^{-1} \cdot q) = (p' \cdot g_0^{},g_0^{-1}
\cdot q)$.
This implies that the action $\Phi_{P \times Q}^{}$ of~$G$ on~%
$P \times Q$ in equation~(\ref{eq:ACTGGPR1}) passes to the quotient
$P \times_{G_0} Q$, and so we get a natural induced action
\begin{equation} \label{eq:ACTGGAB1}
 \begin{array}{cccc}
  \Phi_{P \times_{G_0} Q}^{}:
  & G \times_M (P \times_{G_0} Q) & \longrightarrow & P \times_{G_0} Q
  \\[1mm]
  &        ([p',p],[p,q])         &   \longmapsto   & [p',p] \cdot [p,q] 
 \end{array}
\end{equation}
of the gauge groupoid $\, G = (P \times P)/G_0^{} \,$ on the associated
bundle $P \times_{G_0} Q$, defined by
\begin{equation} \label{eq:ACTGGAB2}
 [p',p] \cdot [p,q]~=~[p',q] \,.
\end{equation}
It will be convenient to visualize this construction in terms of the ``magical
square'' for associated bundles, i.e., the commutative diagram
\begin{equation} \label{eq:MSASSB1}
\begin{array}{c}
\xymatrix{
 ~P \times Q~ \ar[r]^-{\rho_Q^{}} \ar[d]_-{\mathrm{pr}_1} & 
 ~P \times_{G_0} Q~ \ar[d]^-{\,\pi} \\
 ~~~~\vphantom{\hat{P}} P~~~~ \ar[r]_-{\rho\vphantom{M}} &
 ~~~~~M \vphantom{\hat{M}}~~~~~
}
\end{array}
\end{equation}
in which the horizontal projections define principal $G_0^{}$-bundles
while the vertical projections provide fiber bundles with typical
fiber~$Q$ (the first of which is of course just the trivial bundle
over~$P$) such that $\rho_Q^{}$ is an isomorphism on each fiber
and, by definition, is $G$-equivariant.
And again, using the natural projection of~$G$ to the pair groupoid
$M \times M$ of the base manifold~$M$, we get a commutative
diagram:
\begin{equation} \label{eq:ACTGGAB3}
 \begin{array}{c}
  \xymatrix{
   ~G \times_M (P \times_{G_0} Q)~
   \ar[r] \ar[d] & ~P \times_{G_0} Q~ \ar[d] \\
   ~\,(M \times M) \times_M M\,~ \ar[r] & ~~~~~M~~~~~
  }
 \end{array}
\end{equation}
Of course, these actions extend the actions of the gauge group bundle
$P \times_{G_0} G_0^{}$ on the principal bundle~$P$ itself and on the
associated bundle $P \times_{G_0} Q$, respectively, considered in
Ref.~\cite{FS}.

As a first example of induced actions, consider those of the gauge groupoid
of a principal bundle on the vertical bundle of the principal bundle itself
and on the vertical bundle of any of its associated bundles, constructed
according to the prescription specified in equations~(\ref{eq:ACTLG1})--%
(\ref{eq:IAGLJB2}) above.
These actions can be simplified by making use of the fact that the vertical
bundle of a principal bundle is trivial and that the vertical bundle of an
associated bundle is again an associated bundle, i.e., we have canonical
isomorphisms
\begin{equation} \label{eq:VBPB1}
 VP~\cong~P \times \mathfrak{g}_0^{} \,,
\end{equation}
and
\begin{equation} \label{eq:VBAB1}
 V(P \times_{G_0} Q)~\cong~P \times_{G_0} TQ \,,
\end{equation}
both as fiber bundles over~$M$ and as vector bundles over the respective
total spaces~$P$ and $P \times_{G_0} Q$, where in the second case, the
action of~$G_0$ on the tangent bundle $TQ$ of~$Q$ is the one induced
from that on~$Q$.
Similarly, we also have canonical isomorphisms
\begin{equation} \label{eq:LJBPB1}
 \vec{J} P~\cong~L(\pi^*(TM),(P \times \mathfrak{g}_0^{}))~
 \cong~\pi^*(T^* M) \otimes (P \times \mathfrak{g}_0^{}) \,,
\end{equation}
and
\begin{equation} \label{eq:LJBAB1}
 \vec{J} (P \times_{G_0} Q)~\cong~L(\pi^*(TM),P \times_{G_0} TQ)~
 \cong~\pi^*(TM) \otimes (P \times_{G_0} TQ) \,,
\end{equation}
in the same sense.
The statement is then that these bundle isomorphisms are equivariant
under the action of the gauge groupoid $G$, in the first two cases, and
of the Lie groupoid $\, GL(TM) \times_M G$, in the last two cases.

For the proof, we need only consider the statements for the vertical
bundles, since the corresponding ones for the linearized jet bundles
follow directly from them by combining the corresponding actions
of the gauge groupoid with that of the linear frame groupoid $GL(TM)$
of the base manifold~$M$ on the cotangent bundle $T^* M$ of~$M$.
To this end, consider the fundamental vector fields $(X_0^{})_P^{}$
on~$P$ associated to the generators $X_0^{} \in \mathfrak{g}_0^{}$
through the right action of~$G_0^{}$ on~$P$, and for later use, also
the fundamental vector fields $(X_0^{})_Q^{}$ on~$Q$ associated to
the generators $X_0^{} \in \mathfrak{g}_0^{}$ through the left action
of~$G_0^{}$ on~$Q$, defined by
\begin{equation} \label{eq:FVFP}
 (X_0^{})_P^{}(p) \,
 = \, \frac{d}{dt} \bigl( p \cdot \exp(t X_0^{}) \bigr) \Big|_{t=0} \,,
\end{equation}
and by
\begin{equation} \label{eq:FVFQ}
 (X_0^{})_Q^{}(q) \, 
 = \, \frac{d}{dt} \bigl( \exp(- t X_0^{}) \cdot q \bigr) \Big|_{t=0} \,,
\end{equation}
respectively.%
\footnote{We recall that the correspondence in equation~(\ref{eq:FVFP})
establishes a canonical linear isomorphism between the Lie algebra
$\mathfrak{g}_0^{}$ and the vertical space $V_p^{} P$ of~$P$ at~$p$,
whereas the extra minus sign in equation~(\ref{eq:FVFQ}) is introduced
merely for convenience, so as to guarantee consistency of the formulas
when we switch between left and right actions.}
Then the isomorphism in equation~(\ref{eq:VBPB1}) is given by
the mapping that takes the pair $(p,X_0^{})$ to the vertical vector
$(X_0^{})_P^{}(p)$, and that this is equivariant follows immediately
from the following simple calculation:
\[
\begin{aligned} 
 {} [p',p] \cdot (X_0^{})(p)~
 &=~T_p^{} L_{[p',p]}^{} \Bigl(
    \frac{d}{dt} \bigl( p \cdot \exp(t X_0^{}) \bigr) \Big|_{t=0} \Bigr)~
   =~\frac{d}{dt} \bigl( [p',p]
    \cdot (p \cdot \exp(t X_0^{})) \bigr) \, \Big|_{t=0} \\
 &=~\frac{d}{dt} \bigl(
    [p' \cdot \exp(t X_0^{}),p \cdot \exp(t X_0^{})]
    \cdot (p \cdot \exp(t X_0^{})) \bigr) \, \Big|_{t=0} \\
 &=~\frac{d}{dt} \bigl( p' \cdot \exp(t X_0^{}) \bigr) \Big|_{t=0}~
  =~(X_0^{})(p') \,.
\end{aligned}
\]
Similarly, the isomorphism in equation~(\ref{eq:VBAB1}) is
given by the mapping (momentarily denoted by~$\phi$) that takes
$\, [p,\frac{d}{dt} q(t) \big|_{t=0}] \in (P \times_{G_0} TQ)_{[p,q]} \,$
to $\, \frac{d}{dt} [p,q(t)] \big|_{t=0} \in V_{[p,q]}(P \times_{G_0} Q)$,
and that this is equivariant follows immediately from the following simple
calculation:
\[
\begin{aligned} 
 {} [p',p] \cdot
 \phi \bigl( \bigl[ p , \frac{d}{dt} q(t) \Big|_{t=0} \bigr] \bigr)~
 &=~[p',p] \cdot \Bigl( \frac{d}{dt} \, [p,q(t)] \, \Big|_{t=0} \Bigr)~
  =~T_p^{} L_{[p',p]}^{} \Bigl(
    \frac{d}{dt} \, [p,q(t)] \, \Big|_{t=0} \Bigr) \\
 &=~\frac{d}{dt} \bigl( [p',p] \cdot [p,q(t)] \bigr) \, \Big|_{t=0}~
  =~\frac{d}{dt} \, [p',q(t)] \, \Big|_{t=0}~
  =~\phi \bigl( \bigl[ p' , \frac{d}{dt} q(t) \Big|_{t=0} \bigr] \bigr) \\
 &=~\phi \bigl( [p',p] \cdot
                \bigl[ p , \frac{d}{dt} q(t) \Big|_{t=0} \bigr] \bigr) \,.
\end{aligned}
\]

Similar simplifications occur for the other induced actions considered in the
previous section, and this will be discussed in the next two subsections.

\subsection{First order jet groupoids and induced actions}

To begin with, we apply the general procedure developed in Ref.~\cite{CFP}
of ``differentiating'' actions of Lie groupoids on fiber bundles to the
natural actions of the gauge groupoid $\, G = (P \times P)/G_0^{}$
on the principal bundle~$P$ itself and on any associated bundle
$P \times_{G_0} Q$ to obtain natural induced actions
\begin{equation} \label{eq:IAJGJBPB1}
 \begin{array}{cccc}
  \Phi_{JP}^{}:
  & JG \times_M JP & \longrightarrow &       JP
  \\[1mm]
  &   (u_{[p',p]},u_p^{})    &   \longmapsto   & u_{[p',p]} \cdot u_p^{} 
 \end{array}
\end{equation}
and
\begin{equation} \label{eq:IAJGJBAB1}
 \begin{array}{cccc}
  \Phi_{J(P \times_{G_0} Q)}^{}:
  & JG \times_M J(P \times_{G_0} Q) & \longrightarrow
  & J(P \times_{G_0} Q)
  \\[1mm]
  &         (u_{[p',p]},u_{[p,q]})         &   \longmapsto   
  & u_{[p',p]} \cdot u_{[p,q]} 
 \end{array}
\end{equation}
derived from the actions $\Phi_P^{}$ in equation~(\ref{eq:ACTGGPB1})
and $\Phi_{P \times_{G_0} Q}^{}$ in equation~(\ref{eq:ACTGGAB1})
by applying the general formula in equation~(\ref{eq:IAJ1GJ1B2}) of
the previous section.

A more profound understanding of the situation can be obtained by
extending the ``magical square'' for associated bundles in equation~%
(\ref{eq:MSASSB1}) to the corresponding jet bundles, considering the
commutative diagram 
\begin{equation} \label{eq:MSASSB2}
\begin{array}{c}
\xymatrix{
 ~J(P \times Q)~
 \ar[r]^-{J\rho_Q^{}} \ar[d]_-{\pi_{J(P \times Q)}^{}} & 
 ~J(P \times_{G_0} Q)~
 \ar[d]^-{\,\pi_{J(P \times_{G_0} Q)}^{\vphantom{Q}}} \\
 ~~~P \times Q~~~ \ar[r]^-{\rho_Q^{}} \ar[d]_-{\mathrm{pr}_1} & 
 ~~~P \times_{G_0} Q~~~ \ar[d]^-{\,\pi} \\
 \qquad \vphantom{\hat{P}} P \qquad \ar[r]_-{\rho\vphantom{M}} &
 \qquad~ M \vphantom{\hat{M}} ~\qquad
}
\end{array}
\end{equation}
and noting that, just like there is a natural action of~$G$ on~$P \times Q$
derived from that on~$P$ such that $\rho_Q^{}$ is an isomorphism on each
fiber and is $G$-equivariant, as discussed in the previous subsection, there is
also a natural action of~$JG$ on~$J(P \times Q)$ derived from that on~$JP$
such that $J\rho_Q^{}$, although no longer an isomorphism on each fiber
(it is still onto but has a kernel), is $JG$-equivariant.%
\footnote{Note that here, $J(P \times Q)$ is meant to be the jet bundle
of~$P \times Q$ as a bundle over~$M$, i.e., with respect to the projection
$\rho \,\smcirc\, \mathrm{pr}_1^{}$, whereas the previous statement that
$P \times Q$ is a trivial bundle refers to its structure as a bundle over~$P$,
i.e., to the projection $\mathrm{pr}_1^{}$.}

To prove these statements, let us pick points $p \in P$ and $q \in Q$ with
$\, \rho(p) = x \,$ and take tangent maps to the commutative diagram in equation~(\ref{eq:MSASSB1}) to obtain the commutative diagram
\begin{equation} \label{eq:MSASSB3}
 \begin{array}{c}
  \xymatrix{
   ~T_p^{} P \oplus T_q^{} Q~
   \ar[rr]^-{T_{(p,q)} \rho_Q^{}} \ar[d]_{\mathrm{pr}_1}
   && ~T_{[p,q]} (P \times_{G_0} Q)~ \ar[d]^-{\,T_{[p,q]} \pi} \\
   ~~~~~\, \vphantom{\hat{P}} T_p P \,~~~~~
   \ar[rr]_{T_p^{\vphantom{M}} \rho}
   && \qquad~\, T_x M \vphantom{\hat{M}} \,~\qquad
  }
 \end{array}
\end{equation}
Since $\rho_Q^{}$ is a submersion and hence its tangent maps
are surjective, this means that the tangent spaces $T_{[p,q]}%
(P \times_{G_0} Q)$ of the orbit space $P \times_{G_0} Q$ can
be realized as quotient spaces, namely, the linear maps
\begin{equation} \label{eq:TSASSB1}
 T_{(p,q)} \rho_Q^{}: T_p^{} P \oplus T_q^{} Q~~
 \longrightarrow~~T_{[p,q]}(P \times_{G_0} Q)
\end{equation}
induce isomorphisms
\begin{equation} \label{eq:TSASSB2}
 T_{[p,q]}(P \times_{G_0} Q)~
 \cong~(T_p^{} P \oplus T_q^{} Q) / \ker T_{(p,q)} \rho_Q^{} \,,
\end{equation}
and noting that
\begin{equation} \label{eq:JSPROD1}
 J_{(p,q)}(P \times Q)~=~J_p^{} P \oplus L(T_x^{} M,T_q^{} Q) \,,
\end{equation}
this leads to an analogous realization of the jet spaces $J_{[p,q]}%
(P \times_{G_0} Q)$ of the orbit space $P \times_{G_0} Q$ as
quotient spaces, namely, the affine maps
\begin{equation} \label{eq:JSASSB1}
 J_{(p,q)} \rho_Q^{}: J_p^{} P \oplus L(T_x^{} M,T_q^{} Q)~~
 \longrightarrow~~J_{[p,q]}(P \times_{G_0} Q)
\end{equation}
defined by
\begin{equation} \label{eq:JSASSB2}
 J_{(p,q)} \rho_Q^{} (u_p^{},u_q^{})~
 =~ T_{(p,q)} \rho_Q^{} \,\smcirc\, (u_p^{},u_q^{})
\end{equation}
induce isomorphisms
\begin{equation} \label{eq:JSASSB3}
 J_{[p,q]}(P \times_{G_0} Q)~
 \cong~(J_p^{} P \oplus L(T_x^{} M,T_q^{} Q)) \,/\,
       L(T_x^{} M,\ker T_{(p,q)} \rho_Q^{}) \,.
\end{equation}
Now using the $G$-equivariance of~$\rho_Q^{}$, which means
that $\, \Phi_{P \times_{G_0} Q}^{} \,\smcirc\, (\mathrm{id}_G^{}
\times_M^{} \rho_Q^{}) = \rho_Q^{} \,\smcirc\, \Phi_{P \times Q}^{}
\linebreak = \rho_Q^{} \,\smcirc\, (\Phi_P^{} \times \mathrm{id}_Q^{})$
(where in the last equality we have applied the identity $\, G \times_M
(P \times Q) \linebreak = (G \times_M P) \times Q$), we can prove the
$JG$-equivariance of~$J\rho_Q^{}$.
To this end, let us also pick a point $[p',p] \in G$, a jet $\, u_{[p',p]}
\in J_{[p',p]} G \,$ and another jet $u_p^{} \in J_p^{} P$ together
with a linear map $\, u_q^{} \in L(T_x^{} M,T_q^{} Q)$, and
calculate
\vspace{-1ex}
\[
\begin{aligned}
 &u_{[p',p]} \cdot J_{(p,q)} \rho_Q^{} (u_p^{},u_q^{}) \\[1ex]
 & \quad =~T_{([p',p],[p,q])} \Phi_{P \times_{G_0} Q} \,\smcirc\,
           \bigl( u_{[p',p]} \,,\, T_{(p,q)} \rho_Q^{} \,\smcirc\,
                  (u_p^{},u_q^{}) \bigr) \,\smcirc\,
           \pi_{JG}^{\mathrm{fr}}(u_{[p',p]})^{-1} \\[1ex]
 & \quad =~T_{([p',p],[p,q])} \Phi_{P \times_{G_0} Q} \,\smcirc\,
           T_{([p',p],(p,q))} (\mathrm{id}_G^{} \times_M^{} \rho_Q^{})
           \,\smcirc\, \bigl( u_{[p',p]} , (u_p^{},u_q^{}) \bigr)
           \,\smcirc\, \pi_{JG}^{\mathrm{fr}}(u_{[p',p]})^{-1} \\[1ex]
 & \quad =~T_{(p',q)} \rho_Q^{} \,\smcirc\,
           T_{([p',p],(p,q))} \Phi_{P \times Q}^{} \,\smcirc\,
           \bigl( u_{[p',p]} , (u_p^{},u_q^{}) \bigr) \,\smcirc\,
           \pi_{JG}^{\mathrm{fr}}(u_{[p',p]})^{-1} \\[1ex]
 & \quad =~J_{(p',q)} \rho_Q^{}
           \bigl( u_{[p',p]} \cdot (u_p^{},u_q^{}) \bigr)~
         =~J_{(p',q)} \rho_Q^{}
           \bigl( u_{[p',p]} \cdot u_p^{} \,,\, u_q^{} \bigr) \,.
\end{aligned}
\vspace{1ex}
\]
For later use, we also note that
\begin{equation} \label{eq:TSASSB3}
 \ker T_{(p,q)} \rho_Q^{}~
 =~\{ ((X_0^{})_P^{}(p),(X_0^{})_Q^{}(q)) \, | \,
      X_0^{} \in \mathfrak{g}_0^{} \}~
 \cong~\mathfrak{g}_0^{}~\cong~V_p P \,,               
\vspace{1mm}
\end{equation}
where $(X_0^{})_P^{}$ and $(X_0^{})_Q^{}$ denote the fundamental
vector fields on~$P$ and on~$Q$ associated to a generator $\, X_0^{} \in
\mathfrak{g}_0^{} \,$ via the pertinent actions of~$G_0^{}$, respectively,
as defined in equations~(\ref{eq:FVFP}) and~(\ref{eq:FVFQ}) above.
Moreover, under the projection $T_{(p,q)} \rho_Q$, the vertical spaces of
the principal bundle~$P$ and of the associated bundle $P \times_{G_0} Q$
are related by
\begin{equation} \label{eq:VSASSB1}
 V_{[p,q]}^{}(P \times_{G_0} Q)~
 \cong~(V_p^{} P \oplus T_q^{} Q)/\ker T_{(p,q)} \rho_Q^{} \,,
\end{equation}
while, with respect to any principal connection in~$P$ and its associated
connection in~$P \times_{G_0} Q$, the corresponding horizontal spaces of
the principal bundle~$P$ and of the associated bundle $P \times_{G_0} Q$
are related by
\begin{equation} \label{eq:HSASSB1}
 H_{[p,q]}^{}(P \times_{G_0} Q)~
 \cong~(H_p^{} P \oplus \{0\})/\ker T_{(p,q)} \rho_Q^{} \,.
\end{equation}
At the end of this subsection, we shall see how to express the
correspondence between principal connections in~$P$ and their
associated connections in~$P \times_{G_0} Q$ in terms of jets.

Another important property of the action of~$JG$ on~$JP$ in
equation~(\ref{eq:IAJGJBPB1}) is that it commutes with the
right action of the structure group~$G_0^{}$ on~$JP$: this
is essentially obvious because they are induced from an action
of~$G$ on~$P$ and a right action of~$G_0^{}$ on~$P$ which
commute.
But since this is an important fact, let us give a quick formal
proof of the pertinent formula,
\begin{equation}
 u_{[p',p]} \cdot (w_p^{} \cdot g_0^{})~
 =~(u_{[p',p]} \cdot w_p^{}) \cdot g_0^{} \,.
\end{equation}
Indeed, according to equations~(\ref{eq:IAJ1GJ1B2}) and~%
(\ref{eq:ACTGGPB5}) (the second of which can be reformulated
as stating that $\, \Phi_P^{} \,\smcirc\, (\mathrm{id}_G \times
R_{g_0^{}}) = R_{g_0^{}} \,\smcirc\, \Phi_P^{}$, where
$R_{g_0^{}}$ denotes right translation by~$g_0^{}$ in~$P$),
\[
\begin{aligned}
 u_{[p',p]} \cdot (w_p^{} \cdot g_0^{})~
 &=~T_{([p',p],p \cdot g_0)} \Phi_P^{} \,\smcirc\,
   (u_{[p',p]},T_p^{} R_{g_0^{}} \,\smcirc\, w_p^{}) \,\smcirc\,
   \pi_{JG}^{\mathrm{fr}}(u_{[p',p]})^{-1} \\[1ex]
 &=~T_{([p',p],p \cdot g_0)} \Phi_P^{} \,\smcirc\,
    T_{([p',p],p)} (\mathrm{id}_G \times_M^{} R_{g_0^{}})
    \,\smcirc\, (u_{[p',p]},w_p^{}) \,\smcirc\,
    \pi_{JG}^{\mathrm{fr}}(u_{[p',p]})^{-1} \\[1ex]
 &=~T_p^{} R_{g_0^{}} \,\smcirc\, T_{([p',p],p)} \Phi_P^{}
   \,\smcirc\, (u_{[p',p]},w_p^{}) \,\smcirc\,
   \pi_{JG}^{\mathrm{fr}}(u_{[p',p]})^{-1} \\[1ex]
 &=~(u_{[p',p]} \cdot w_p^{}) \cdot g_0^{} \,.
\end{aligned}
\]
This implies that the action $\Phi_{JP}^{}$ of~$JG$ on~$JP$ in
equation~(\ref{eq:IAJGJBPB1}) passes to the quotient
\begin{equation} \label{eq:DEFCB}
 CP~=~JP/G_0^{} \,,
\end{equation}
which is an affine bundle over~$M$ called the \emph{connection
bundle} of~$P$ because its sections correspond precisely to the
$G_0^{}$-equivariant sections of~$JP$ (as an affine bundle
over~$P$), which are exactly the principal connections on~$P$.
Thus we get a natural induced action
\begin{equation} \label{eq:IAJGCB1}
 \begin{array}{cccc}
  \Phi_{CP}^{}:
  &     JG \times_M CP    & \longrightarrow &            CP
  \\[1mm]
  & (u_{[p',p]},[w_p^{}]) &   \longmapsto   & u_{[p',p]} \cdot [w_p^{}]
 \end{array}
\end{equation}
of~$JG$ on~$CP$.
It will be convenient to visualize this construction in terms of the
``magical square'' for connection bundles, i.e., the commutative
diagram
\begin{equation} \label{eq:MSCONB1}
\begin{array}{c}
\xymatrix{
 ~JP_{\vphantom{p}}~ \ar[r]^-{\rho_C^{}} \ar[d]_-{\pi_{JP}^{}} &
 ~CP_{\vphantom{p}}~ \ar[d]^-{\,\pi_{CP}^{}} \\
 ~~\vphantom{\hat{P}} P~~ \ar[r]_-{\rho} &
 ~\,M \vphantom{\hat{M}}\,~
}
\end{array}
\end{equation}
in which the horizontal projections define principal $G_0^{}$-bundles
while the vertical projections provide affine bundles such that $\rho_C^{}$
is an isomorphism on each fiber and, by definition, is $JG$-equivariant.

Now we can formulate the rule that to each principal connection in~$P$
assigns its associated connection in $P \times_{G_0} Q$ in terms of a
canonical bundle map over~$P \times_{G_0} Q$, namely:
\begin{equation} \label{eq:ASSCON1}
 \begin{array}{ccc}
  \pi^*(CP) & \longrightarrow &       J(P \times_{G_0} Q)      \\[1mm]
  ([p,q],[w_p^{}])  &   \longmapsto   & J_{(p,q)} \rho_Q^{} (w_p^{},0)
 \end{array}
\end{equation}
To see that it is well defined, we have to check that, given
any point $x \in M$, the result remains unchanged if we pick
any $g_0^{} \in G_0^{}$ to replace the representative
$\, (p,q) \in (P \times Q)_x^{} \,$ of \linebreak $[p,q] \in
(P \times_{G_0} Q)_x^{} \,$ by another representative
$\, (p \cdot g_0^{}, g_0^{-1} \cdot q) \,$ and the representative
$\, w_p^{} \in J_p^{} P \,$ of $\, [w_p] \in C_x^{} P \,$ by another
representative $\, w_{p \cdot g_0}^{} \,$: writing $R_{g_0^{}}^P$
for right translation by $g_0^{}$ in~$P$ and $L_{g_0^{-1}}^Q$ for
left translation by $g_0^{-1}$ in~$Q$, we have $\, w_{p \cdot g_0}^{}
= T_p^{} R_{g_0^{}}^P \,\smcirc\, w_p^{} \,$ and get
\vspace{-1ex}
\[
\begin{aligned}
 J_{(p \cdot g_0^{},g_0^{-1} \cdot q)} \rho_Q^{}
 (w_{p \cdot g_0}^{},0)~
 &=~T_{(p \cdot g_0^{},g_0^{-1} \cdot q)} \rho_Q^{} \,\smcirc\,
   \bigl( T_p^{} R_{g_0^{}}^P \,\smcirc\, w_p^{} \,,\, 0 \bigr) \\
 &=~T_{(p \cdot g_0^{},g_0^{-1} \cdot q)} \rho_Q^{} \,\smcirc\,
   T_{(p,q)} \bigl( R_{g_0^{}}^P \times L_{g_0^{-1}}^Q \bigr)
   \,\smcirc\, (w_p^{},0) \\
 &=~T_{(p,q)} \rho_Q^{} \,\smcirc\, (w_p^{},0)~
  =~J_{(p,q)} \rho_Q^{} (w_p^{},0) \,.
\end{aligned}
\]
Moreover, this bundle map is also $JG$-equivariant: this follows trivially
from the definition of the action of~$JG$ on the spaces involved and the
$JG$-equivariance of $J\rho_Q^{}$ that was proved above.
And finally, we observe that this bundle map does capture the essence
of passing from a principal connection to its associated connection, since
if the former is given by a section $\, \varGamma^P: M \longrightarrow CP \,$
and the latter by a section $\, \varGamma^{P \times_{G_0} Q}:
P \times_{G_0} Q \longrightarrow J(P \times_{G_0} Q)$, then
$\, \varGamma^{P \times_{G_0} Q}$ is simply the push-forward
of the section $\, \varGamma^P \smcirc\, \pi: P \times_{G_0} Q
\longrightarrow \pi^*(CP) \,$ with this bundle map.
Note also that the prescription corresponds precisely to that given
in equation~(\ref{eq:HSASSB1}) at the level of horizontal bundles.

\subsection{Second order jet groupoids and induced actions}

In this subsection, we apply the general procedure developed in
Ref.~\cite{CFP} of ``differentiating'' actions of Lie groupoids on
fiber bundles once more, namely, to the natural actions of the jet
groupoid $JG$ of the gauge groupoid $\, G = (P \times P)/G_0^{}$
on the jet bundle~$JP$ and the connection bundle~$CP$ of the
principal bundle~$P$ itself, to obtain natural induced actions\,%
\footnote{In this subsection, we often write $g=[p',p]$ for points
in the gauge groupoid $G = (P \times P)/G_0^{}$.}
\begin{equation} \label{eq:IAJJGJJBPB1}
 \begin{array}{cccc}
  \Phi_{J(JP)}:
  & J(JG) \times_M J(JP) & \longrightarrow &            J(JP)
  \\[1mm]
  & (u'_{u_g},u'_{u_p})  &   \longmapsto   & u'_{u_g} \!\cdot\, u'_{u_p}
 \end{array}
\end{equation}
and
\begin{equation} \label{eq:IAJJGJBCB1}
 \begin{array}{cccc}
  \Phi_{J(CP)}:
  &     J(JG) \times_M J(CP)      & \longrightarrow &
  J(CP) \\[1mm]
  & (u_{u_g}^{\prime},u_{[w_p]}^{}) &   \longmapsto   &
  u_{u_g}^{\prime} \!\cdot\, u_{[w_p]}^{}
 \end{array}
\end{equation}
derived from the actions $\Phi_{JP}^{}$ in equation~(\ref{eq:IAJGJBPB1})
and $\Phi_{CP}^{}$ in equation~(\ref{eq:IAJGCB1}) by applying the general
formula in equation~(\ref{eq:IAJ1GJ1B2}) of the previous section.
Explicitly, we have
\begin{equation} \label{eq:IAJJGJJBPB2}
 u'_{u_g} \!\cdot\, u'_{u_p}~
 =~T_{(u_g,u_p)} \Phi_{JP}^{} \,\smcirc\, (u'_{u_g},u'_{u_p})
   \,\smcirc\, \pi_{J(JG)}^{\mathrm{fr}}(u'_{u_g})^{-1} \,,
\end{equation}
and
\begin{equation} \label{eq:IAJJGJBCB2}
 u'_{u_g} \!\cdot\, u'_{[w_p]}~
 =~T_{(u_g,[w_p])} \Phi_{CP}^{} \,\smcirc\, (u'_{u_g},u'_{[w_p]})
   \,\smcirc\, \pi_{J(JG)}^{\mathrm{fr}}(u'_{u_g})^{-1} \,,
\end{equation}
respectively.
These actions admit restrictions to several subgroupoids and subbundles,
among which the following will become important to us at some point or
another: the natural induced actions
\begin{equation} \label{eq:IASJGSJBPB1}
 \begin{array}{cccc}
  \Phi_{\bar{J}^{\>\!2} P}:
  & \bar{J}^{\>\!2} G \times_M \bar{J}^{\>\!2} P
  & \longrightarrow & \bar{J}^{\>\!2} P \\[1mm]
  & (u'_{u_g},u'_{u_p})  &   \longmapsto   & u'_{u_g} \!\cdot\, u'_{u_p}
 \end{array}
\end{equation}
of the semiholonomous second order jet groupoid~$\bar{J}^{\>\!2} G$
of~$G$ and
\begin{equation} \label{eq:IAJ2GSJBPB1}
 \begin{array}{cccc}
  \Phi_{\bar{J}^{\>\!2} P}:
  & J^{\>\!2} G \times_M \bar{J}^{\>\!2} P
  & \longrightarrow & \bar{J}^{\>\!2} P \\[1mm]
  & (u'_{u_g},u'_{u_p})  &   \longmapsto   & u'_{u_g} \!\cdot\, u'_{u_p} 
 \end{array}
\end{equation}
of the second order jet groupoid~$J^{\>\!2} G$ of~$G$ on the
semiholonomous second order jet bundle $\bar{J}^{\>\!2} P$ of~$P$,
as well as the action
\begin{equation} \label{eq:IAJ2GJ2BPB1}
 \begin{array}{cccc}
  \Phi_{J^{\>\!2} P}:
  & J^{\>\!2} G \times_M J^{\>\!2} P
  & \longrightarrow & J^{\>\!2} P \\[1mm]
  & (u'_{u_g},u'_{u_p})  &   \longmapsto   & u'_{u_g} \!\cdot\, u'_{u_p}
 \end{array}
\end{equation}
of the second order jet groupoid~$J^{\>\!2} G$ of~$G$ on the
second order jet bundle $J^{\>\!2} P$ of~$P$, all defined by the
same formula,
\begin{equation} \label{eq:IASJGSJBPB2}
 u'_{u_g} \!\cdot\, u'_{u_p}~
 =~T_{(u_g,u_p)} \Phi_{JP}^{} \,\smcirc\, (u'_{u_g},u'_{u_p})
   \,\smcirc\, \pi_{JG}^{\mathrm{fr}}(u_g^{})^{-1} \,,
\end{equation}
and similarly, the natural induced actions
\begin{equation} \label{eq:IASJGJBCB1}
 \begin{array}{cccc}
  \Phi_{J(CP)}:
  & \bar{J}^{\>\!2} G \times_M J(CP) & \longrightarrow &
  J(CP) \\[1mm]
  & (u_{u_g}^{\prime},u_{[w_p]}^{})  &   \longmapsto   &
  u_{u_g}^{\prime} \!\cdot\, u_{[w_p]}^{}
 \end{array}
\end{equation}
of the semiholonomous second order jet groupoid $\bar{J}^{\>\!2} G$
of~$G$ and
\begin{equation} \label{eq:IAJ2GJBCB1}
 \begin{array}{cccc}
  \Phi_{J(CP)}:
  &  J^{\>\!2} G \times_M J(CP)   & \longrightarrow &
  J(CP) \\[1mm]
  & (u_{u_g}^{\prime},u_{[w_p]}^{}) &   \longmapsto   &
  u_{u_g}^{\prime} \!\cdot\, u_{[w_p]}^{}
 \end{array}
\end{equation}
of the second order jet groupoid~$J^{\>\!2} G$ of~$G$ on the jet
bundle $J(CP)$ of the connection bundle~$CP$ of~$P$, defined by
\begin{equation} \label{eq:IASJGJBCB2}
 u'_{u_g} \cdot u_{[w_p]}^{}~
 =~T_{(u_g,[w_p])} \Phi_{CP}^{} \,\smcirc\, (u'_{u_g},u_{[w_p]}^{})
   \,\smcirc\, \pi_{JG}^{\mathrm{fr}}(u_g^{})^{-1} \,.
\end{equation}
As noted in the discussion preceding Proposition~\ref{prp:EQUIV2} in
the previous section, the simplification in the last term on the rhs of
equations~(\ref{eq:IASJGSJBPB2}) and~(\ref{eq:IASJGJBCB2}), as
compared to equations~(\ref{eq:IAJJGJJBPB2}) and~(\ref{eq:IAJJGJBCB2}),
comes from the assumption that $u'_{u_g}$ is semiholonomous, and the
definition of the actions in equations~(\ref{eq:IASJGSJBPB1}) and~%
(\ref{eq:IAJ2GJ2BPB1}) relies on the fact that when $u'_{u_g}$ and
$u'_{u_p}$ are both semiholonomous or both holonomous, then so is
$u'_{u_g} \!\cdot\, u'_{u_p}$.

A more profound understanding of the situation can be obtained by
extending the ``magical square'' for connection bundles in equation~%
(\ref{eq:MSCONB1}) to the corresponding jet bundles, considering the
commutative diagram 
\begin{equation} \label{eq:MSCONB2}
 \begin{array}{c}
  \xymatrix{
   ~J(JP)_{\vphantom{p}}~
   \ar[r]^-{J\rho_C^{}} \ar[d]_-{\pi_{J(JP)}^{}} &
   ~J(CP)_{\vphantom{p}}~
   \ar[d]^-{\,\pi_{J(CP)}^{}} \\
   ~~~JP\vphantom{\hat{P}_p}~~~
   \ar[r]^-{\rho_C^{}} \ar[d]_-{\pi_{JP}^{}} &
   ~~~\,CP\vphantom{\hat{P}_p}\,~~~
   \ar[d]^-{\,\pi_{CP}^{}} \\
   ~~~~\vphantom{\hat{P}} P~~~~
   \ar[r]^-{\rho} &
   ~~~~M \vphantom{\hat{M}}~~~~
  }
 \end{array}
\end{equation}
and noting that $J\rho_C^{}$, although no longer an isomorphism on
each fiber (it is still onto but has a kernel), is $J(JG)$-equivariant.
Even more importantly, by restricting to the semiholonomous second
order jet bundle of~$P$, we arrive at a ``magical square'' for jet
bundles of connection bundles, i.e., the commutative diagram
\begin{equation} \label{eq:MSJCNB1}
 \begin{array}{c}
  \xymatrix{
   ~\bar{J}^{\>\!2} P_{\vphantom{p}}~
   \ar[r]^-{J\rho_C^{}} \ar[d]_-{\pi_{\bar{J}^{\>\!2} P}} &
   ~J(CP)_{\vphantom{p}}~
   \ar[d]^-{\,\pi_{J(CP)}^{\vphantom{2}}} \\
   ~\;JP\vphantom{\hat{P}_p}\;~
   \ar[r]^-{\rho_C^{}} \ar[d]_-{\pi_{JP}^{}} &
   ~~~CP\vphantom{\hat{P}_p}~~~
   \ar[d]^-{\,\pi_{CP}^{}} \\
   ~~\;\vphantom{\hat{P}} P\;~~
   \ar[r]^-{\rho} &
   ~~~~M \vphantom{\hat{M}}~~~~
  }
 \end{array}
\end{equation}
in which all three horizontal projections define principal $G_0^{}$-bundles
while the vertical projections provide affine bundles such that $\rho_C^{}$
and $J\rho_C^{}$ are both isomorphisms on each fiber, $\rho_C^{}$ is
$JG$-equivariant and $J\rho_C^{}$ is $\bar{J}^{\>\!2} G$-equivariant.

To prove these statements, let us pick a point $p \in P$ with $\, \rho(p)
= x \,$ and a jet $w_p^{} \in J_p^{} P$ and take tangent maps to the
commutative diagram in equation~(\ref{eq:MSCONB1}) to obtain the
commutative diagram
\begin{equation} \label{eq:MSCONB3}
 \begin{array}{c}
  \xymatrix{
   ~T_{w_p}^{}(JP)~
   \ar[rr]^-{T_{w_p}^{} \rho_C^{}} \ar[d]_-{T_{w_p}^{} \pi_{JP}^{}}
   && ~T_{[w_p]}(CP)~
   \ar[d]^-{\,T_{[w_p]} \pi_{CP}^{}} \\
   \quad~ T_p^{\vphantom{P}} P ~\quad
   \ar[rr]_{T_p^{\vphantom{M}} \rho}
   && \quad~ T_x^{\vphantom{M}} M ~\quad
  }
 \end{array}
\end{equation}
Since $\rho_C^{}$ is a submersion and hence its tangent maps are
surjective, this means that the tangent spaces $T_{[w_p]}(CP)$
of the orbit space $CP$ can be realized as quotient spaces, namely,
the linear maps
\begin{equation} \label{eq:TSCONB1}
 T_{w_p}{} \rho_C^{}: T_{w_p}^{}(JP)~~
 \longrightarrow~~T_{[w_p]}(CP)
\end{equation}
induce isomorphisms
\begin{equation} \label{eq:TSCONB2}
 T_{[w_p]}^{}(CP)~
 \cong~T_{w_p}^{}(JP) / \ker T_{w_p}{} \rho_C^{} \,,
\end{equation}
and this leads to an analogous realization of the jet spaces $J_{[w_p]}(CP)$
of the orbit space $CP$ as quotient spaces, namely, the affine maps
\begin{equation} \label{eq:JSCONB1}
 J_{w_p}{} \rho_C^{}: J_{w_p}^{}(JP)~~
 \longrightarrow~~J_{[w_p]}(CP)
\end{equation}
defined by
\begin{equation} \label{eq:JSCONB2}
 J_{w_p} \rho_C^{} (u'_{w_p})~
 =~ T_{w_p}^{} \rho_C^{} \,\smcirc\, u'_{w_p}
\end{equation}
induce isomorphisms
\begin{equation} \label{eq:JSCONB3}
 J_{[w_p]}(CP)~
 \cong~J_{w_p}^{}(JP) \,/\, L(T_x^{} M,\ker T_{w_p}^{} \rho_C^{}) \,.
\end{equation}
Now using the $JG$-equivariance of~$\rho_C^{}$, which means
that $\, \Phi_{CP}^{} \,\smcirc\, (\mathrm{id}_{JG}^{} \times_M^{}
\rho_C^{}) = \rho_C^{} \,\smcirc\, \Phi_{JP}^{}$, we can prove the
$J(JG)$-equivariance of~$J\rho_C^{}$.
To this end, let us also pick a point $g = [p',p] \in G$ and a jet $u_g^{}
\in J_g^{} G$, together with iterated jets $u'_{u_g} \in J_{u_g}^{}(JG)$
and $u'_{w_p} \in J_{w_p}^{}(JP)$, and calculate
\vspace{-1ex}
\[
\begin{aligned}
 &u'_{u_g} \cdot J_{w_p}^{} \rho_C^{} (u'_{w_p}) \\[1ex]
 & \quad =~T_{(u_g,[w_p])} \Phi_{CP}^{} \,\smcirc\,
           \bigl( u'_{u_g} \,,\, T_{w_p}^{} \rho_C^{} \,\smcirc\,
                  u'_{w_p} \bigr) \,\smcirc\,
           \pi_{J(JG)}^{\mathrm{fr}}(u'_{u_g})^{-1} \\[1ex]
 & \quad =~T_{(u_g,[w_p])} \Phi_{CP}^{} \,\smcirc\,
           T_{(u_g,w_p)} (\mathrm{id}_{JG}^{} \times_M^{} \rho_C^{})
           \,\smcirc\, \bigl( u'_{u_g} , u'_{w_p} \bigr) \,\smcirc\,
           \pi_{J(JG)}^{\mathrm{fr}}(u'_{u_g})^{-1} \\[1ex]
 & \quad =~T_{u_g \cdot w_p}^{} \rho_C^{} \,\smcirc\,
           T_{(u_g,w_p)} \Phi_{JP}^{} \,\smcirc\,
           \bigl( u'_{u_g} , u'_{w_p} \bigr) \,\smcirc\,
           \pi_{J(JG)}^{\mathrm{fr}}(u'_{u_g})^{-1} \\[1ex]
 & \quad =~J_{u_g \cdot w_p}^{} \rho_C^{}
           \bigl( u'_{u_g} \cdot u'_{u_p} \bigr) \,.
\end{aligned}
\vspace{1ex}
\]
But here we can actually do better if we replace iterated jets by
semiholonomous second order jets because that will eliminate the need
of passing to a quotient and convert the commutative diagram in
equation~(\ref{eq:MSCONB2}) to the one in equation~(\ref{eq:MSJCNB1}).
To show this, we first note that, as before,
\begin{equation} \label{eq:TSCONB3}
 \ker T_{w_p}^{} \rho_C^{}~
 =~\{ (X_0^{})_{JP}^{}(w_p) \, | \, X_0^{} \in \mathfrak{g}_0^{} \}~
 \cong~\mathfrak{g}_0^{}~\cong~V_p^{} P \,,               
\end{equation}
where $(X_0^{})_{JP}^{}$ denotes the fundamental vector field on~$JP$
associated to a generator $\, X_0^{} \in \mathfrak{g}_0^{} \,$ via the
pertinent action of~$G_0^{}$, defined by the appropriate analogue of
equation~(\ref{eq:FVFP}) above.
Here, we shall need a more explicit form of this isomorphism between
the spaces $\, \ker T_{w_p}^{} \rho_C^{} \,$ and $V_p P$: it is simply
the restriction
\begin{equation} \label{eq:TSCONB4}
 T_{w_p}^{} \pi_{JP}^{}: \ker T_{w_p}^{} \rho_C^{}~~
 \stackrel{\cong}{\longrightarrow}~~V_p^{} P
\end{equation}
of the linear map
\begin{equation} \label{eq:TSCONB5}
 T_{w_p}^{} \pi_{JP}^{}: T_{w_p}^{}(JP)~~
 \longrightarrow~~T_p^{} P
\end{equation}
that appears in the definition of semiholonomous second order jets.
(Indeed, the right action of~$G_0^{}$ on~$JP$ being induced from
that on~$P$, the tangent map $T_{w_p}^{} \pi_{JP}^{}$ will of
course take any fundamental vector field $(X_0^{})_{JP}^{}$
at~$w_p^{}$ to the corresponding fundamental vector field
$(X_0^{})_P^{}$ at~$p$.)
This in turn implies that the restriction of the (affine) map in equation~%
(\ref{eq:JSCONB1}) to the (affine) subspace $\bar{J}_{w_p}^{\>\!2} P$
of the (affine) space $J_{w_p}^{}(JP)$ will establish an isomorphism
\begin{equation} \label{eq:JSCONB4}
 J_{w_p}{} \rho_C^{}: \bar{J}_{w_p}^{\>\!2} P~~
 \stackrel{\cong}{\longrightarrow}~~J_{[w_p]}(CP)
\end{equation}
so we can replace equation~(\ref{eq:JSCONB3}) by the much simpler
equation
\begin{equation} \label{eq:JSCONB5}
 J_{[w_p]}(CP)~\cong~\bar{J}_{w_p}^{\>\!2} P \,.
\end{equation}
To prove this statement, we have to show that the affine map
in equation~(\ref{eq:JSCONB1}), when restricted to the affine
subspace $\bar{J}_{w_p}^{\>\!2} P$, (a) becomes injective
and (b) remains surjective.
For~(a), assume we are given two semiholonomous second
order jets $\, u_{w_p}^{\prime\,1},u_{w_p}^{\prime\,2}
\in \bar{J}_{w_p}^{\>\!2} P \,$ which under $J_{w_p}^{}
\rho_C^{}$ have the same image; then their difference
is a linear map from $T_x^{} M$ to~$T_{w_p}^{}(JP)$
satisfying two conditions, namely that its composition with
$T_{w_p}^{} \rho_C^{}$ is zero, so it takes value in
$\ker T_{w_p}^{} \rho_C^{}$, and that its composition
with $T_{w_p}^{} \pi_{JP}^{}$ is also zero, since
$u_{w_p}^{\prime\,1}$ and $u_{w_p}^{\prime\,2}$
are both semiholonomous.
But this implies that it must itself be zero since according
to equation~(\ref{eq:TSCONB4}), $T_{w_p}^{} \pi_{JP}^{}$ is
injective on~$\ker T_{w_p}^{} \rho_C^{}$.
For~(b), assume we are given a general iterated jet $\, u'_{w_p}
\in J_{w_p}^{}(JP) \,$ and consider the difference $\, T_{w_p}^{}
\pi_{JP}^{} \,\smcirc\, u'_{w_p} - w_p^{}$, which is a linear
map from $T_x^{} M$ to $V_p^{} P$, so that according to
equation~(\ref{eq:TSCONB4}), there is a unique linear map
$\vec{u}_{w_p}^{\,\prime}$ from $T_x^{} M$ to $\, \ker
T_{w_p}^{} \rho_C^{} \subset T_{w_p}^{}(JP) \,$ satisfying
$\, T_{w_p}^{} \pi_{JP}^{} \,\smcirc\, u'_{w_p} - w_p^{} =
T_{w_p}^{} \pi_{JP}^{} \,\smcirc\, \vec{u}_{w_p}^{\,\prime}$.
But this implies that the difference $\, \bar{u}'_{w_p} = u'_{w_p}
- \, \vec{u}_{w_p}^{\,\prime} \,$ is a semiholonomous second
order jet, $\bar{u}'_{w_p} \in \bar{J}_{w_p}^{\>\!2} P$,
which under $J_{w_p}^{} \rho_C^{}$ has the same image
as the original iterated jet $\, u'_{w_p} \in J_{w_p}^{}(JP)$.

\section{Minimal coupling and Utiyama's theorem II}

In the context of the formalism adopted in the previous section,
the minimal coupling prescription and the curvature map can be
viewed as stemming from bundle maps
\begin{equation} \label{eq:COVDER2}
 D: CP \times_M J(P \times_{G_0} Q)~~\longrightarrow~~
 \vec{J}(P \times_{G_0} Q) \,,
\end{equation}
and
\begin{equation} \label{eq:CURV2}
 F: J(CP)~~\longrightarrow~~
 \bwedge^{\!2\,} T^* M \otimes (P \times_{G_0} \mathfrak{g}_0^{}) \,,
\end{equation}
over~$M$, which have already appeared in Ref.~\cite{FS} (see the
diagrams in equations~(52) and~(57) there).
What we want to show here is that, and in precisely what sense, these
bundle maps are equivariant under the action not only of the pertinent
Lie group bundles but also of the pertinent Lie groupoids.
To this end, it turns out to be convenient to ``lift'' all bundles to the
space appearing in the upper left hand corner of the appropriate
``magical square'', that is, the space $P \times Q$ in the first case
(see equation~(\ref{eq:MSASSB1})) and the space $JP$ in the second
case (see equation~(\ref{eq:MSCONB1})), where these bundle maps
take a much simpler form.

\subsection{Minimal coupling}

To deal with the minimal coupling prescription, we observe that
the bundle map $D$ in equation~(\ref{eq:COVDER2}) fits into the 
following commutative diagram
\begin{equation} \label{eq:COVDER3}
 \begin{array}{c}
  \xymatrix{
   ~(JP \times Q) \times_{P \times Q}^{} J(P \times Q)~
   \ar[r]^-{D}
   \ar[d]_-{(\rho_C^{} \ssmcirc \mathrm{pr}_1^{},J\rho_Q^{})\,} &
   ~~\,\vec{J}(P \times Q)\,~~ \ar[d]^-{\,\vec{J}\rho_Q^{}} \\
   \quad~~ CP \times_M J(P \times_{G_0} Q) ~~\quad
   \ar[r]_-{D^{\vphantom{M}}} &
   ~\vec{J}(P \times_{G_0} Q)~
  }
 \end{array}
\end{equation}
where the bundles in the top row are over $P \times Q$ while those
in the bottom row are over~$M$.
(Here, we have identified the pull-back of~$JP$ by the projection
from~$P \times Q$ to~$P$ with the cartesian product $JP \times Q$.)
In fact, it is convenient to expand this to a commutative diagram
\begin{equation} \label{eq:COVDER4}
 \begin{array}{c}
  \xymatrix{
   ~~\; (JP \times Q) \times_{P \times Q}^{} J(P \times Q) \;~~
   \ar[r]^-{D} \ar[d] &
   ~~\,\vec{J}(P \times Q)\,~~ \ar[d] \\
   ~\pi^*(CP) \times_{P \times_{G_0} Q}^{} J(P \times_{G_0} Q)~
   \ar[r]^-{D} \ar[d] &
   ~\vec{J}(P \times_{G_0} Q)~ \ar[d] \\
   \quad~~~\; CP \times_M J(P \times_{G_0} Q) \;~~~\quad
   \ar[r]^-{D} &
   ~\vec{J}(P \times_{G_0} Q)~
  }
 \end{array}
\end{equation}
where the bundles in the middle row are over the quotient space
$P \times_{G_0} Q$, i.e., the total space of the corresponding
associated bundle.
(Here, we omit the labels on the vertical maps, which are either
the same as in the previous diagram or else are obvious.)
Then the bundle map~$D$ in the middle row is the composition of
the difference map already introduced at the beginning of this paper
(see equation~(\ref{eq:DIFMAP1})) and the canonical bundle map
of equation~(\ref{eq:ASSCON1}) in the first factor, \linebreak
up to a sign that can be taken care~of by switching the two factors.
Continuing to use the same notation as in Section~3.1, we see that
this corresponds to the bundle map~$D$ in the bottom row being given
in terms of that in the top row according to
\begin{equation} \label{eq:COVDER5}
 D_x^{} \bigl( [w_p^{}] , J_{(p,q)} \rho_Q^{}(u_p^{},u_q^{}) \bigr)~
 =~\vec{J}_{(p,q)} \rho_Q^{}
   \bigl( D_{(p,q)} \bigl( w_p^{} , (u_p^{},u_q^{}) \bigr) \bigr) \,,
\end{equation}
whereas the latter is simply defined by
\begin{equation} \label{eq:COVDER6}
 D_{(p,q)} \bigl( w_p^{} , (u_p^{},u_q^{}) \bigr)~
 =~(u_p^{} - w_p^{} , u_q^{}) \,.
\end{equation}
(This follows from equations~(\ref{eq:JSPROD1})--(\ref{eq:JSASSB3})
together with the same equations with $J$ replaced by~$\vec{J}$.)
\linebreak
To show that $D_x^{}$ is well defined, note first that if we replace the
point $p$ in~$P$ by another point in~$P$ in the same fiber over~$x$,
which is of the form $p \cdot g_0^{}$ for some (unique) $g_0^{} \in G$,
then we must replace $w_p^{}$ by $\, w_p^{} \cdot g_0^{} = T_p^{}
R_{g_0^{}} \,\smcirc\, w_p^{} \,$ and similarly $u_p^{}$ by
$\, u_p^{} \cdot g_0^{} = T_p^{} R_{g_0^{}} \,\smcirc\, u_p^{}$,
as well as $u_q^{}$ by $\, u_q^{} \cdot g_0^{} = T_q^{} L_{g_0^{-1}}
\,\smcirc\, u_q^{}$, so as to guarantee that $J_{(p,q)} \rho_Q^{}
(u_p^{},u_q^{})$ remains unaltered:
\[
\begin{aligned}
 &J_{(p \cdot g_0^{},g_0^{-1} \cdot q)} \rho_Q^{}
  \bigl( u_p^{} \cdot g_0^{} , u_q^{} \cdot g_0^{} \bigr)~
  =~T_{(p \cdot g_0^{},g_0^{-1} \cdot q)} \rho_Q^{} \,\smcirc\,
    T_{(p,q)} \bigl( R_{g_0^{}} \times L_{g_0^{-1}} \bigr) \,\smcirc\,
    (u_p^{},u_q^{}) \\[1ex]
 & \quad =~T_{(p,q)} \bigl( \rho_Q^{} \,\smcirc\,
           \bigl( R_{g_0^{}} \times L_{g_0^{-1}} \bigr) \bigr) \,\smcirc\,
           (u_p^{},u_q^{})~
         =~T_{(p,q)} \rho_Q^{} \,\smcirc\, (u_p^{},u_q^{}) \\[1ex]
 & \quad =~J_{(p,q)} \rho_Q^{}(u_p^{},u_q^{}) \,.
\end{aligned}
\]
But then $\vec{J}_{(p,q)} \rho_Q^{}(u_p^{} - w_p^{} , u_q^{})$
will remain unaltered as well:
\[
\begin{aligned}
 &\vec{J}_{(p \cdot g_0^{},g_0^{-1} \cdot q)} \rho_Q^{}
  \bigl( (u_p^{} - w_p^{}) \cdot g_0^{} , u_q^{} \cdot g_0^{} \bigr)~
  =~T_{(p \cdot g_0^{},g_0^{-1} \cdot q)} \rho_Q^{} \,\smcirc\,
    T_{(p,q)} \bigl( R_{g_0^{}} \times L_{g_0^{-1}} \bigr) \,\smcirc\,
    (u_p^{}-w_p^{},u_q^{}) \\[1ex]
 & \quad =~T_{(p,q)} \bigl( \rho_Q^{} \,\smcirc\,
           \bigl( R_{g_0^{}} \times L_{g_0^{-1}} \bigr) \bigr) \,\smcirc\,
           (u_p^{}-w_p^{},u_q^{})~
         =~T_{(p,q)} \rho_Q^{} \,\smcirc\, (u_p^{}-w_p^{},u_q^{}) \\[1ex]
 & \quad =~\vec{J}_{(p,q)} \rho_Q^{}(u_p^{}-w_p^{},u_q^{}) \,.
\end{aligned}
\]
Moreover, even if we leave $p$ fixed, we may still modify the second
component in the argument of $D_{(p,q)}$, i.e., the pair $\, (u_p^{},
u_q^{}) \in J_p^{} P \times L(T_x^{} M,T_q^{} Q)$, without changing
its image under $J_{(p,q)} \rho_Q^{}$, namely, by adding a pair
$\, (\vec{u}_p^{},\vec{u}_q^{}) \in L(T_x^{} M, \ker T_{(p,q)}
\rho_Q^{})$.
But then since $w_p^{} \in J_p^{} P$ remains unaltered,
the expression $\, (u_p^{}-w_p^{},u_q^{}) \in \vec{J}_p^{} P
\times L(T_x^{} M,T_q^{} Q) \,$ will be modified in the same
way and, in particular, without changing its image under
$\vec{J}_{(p,q)} \rho_Q^{}$.

Now we are ready to formulate the first main theorem in this paper,
which extends the left part of the commutative diagram in equation~(52)
of Ref.~\cite{FS}, as follows.
\begin{thm}~
 The minimal coupling map $D$ in equation~(\ref{eq:COVDER2}) is
 equivariant under the actions of the pertinent Lie groupoids, i.e., the
 diagram
 \begin{equation}
  \begin{array}{c}
   \xymatrix{
    ~~~\, JG \times_M (CP \times_M J(P \times_{G_0} Q)) \,~~~
    \ar[r]
    \ar[d]_{(\pi_{JG}^{\mathrm{fr}} \times \pi_{JG}^{}) \times_M^{} D\,} &
    ~CP \times_M J(P \times_{G_0} Q)~ \ar[d]^{\,D} \\
    ~(GL(TM) \times_M G) \times_M \vec{J} (P \times_{G_0} Q)~ \ar[r] &
    \qquad \vec{J} (P \times_{G_0} Q) \qquad
   }
  \end{array}
 \end{equation}
 commutes.
\end{thm}

\begin{proof}
 This follows immediately from equivariance of~$J\rho_Q^{}$ under~%
 $JG$ (which as we have seen implies equivariance of the canonical
 bundle map in equation~(\ref{eq:ASSCON1}) under~$JG$) and
 equivariance of~$\vec{J} \rho_Q^{}$ under $GL(TM) \times_M G$
 (which can be shown in precisely the same way), in combination with
 Proposition~\ref{prp:EQUIV1}, to prove that the bundle maps~$D$
 in the top and middle rows of the diagram in equation~%
 (\ref{eq:COVDER4}) are equivariant in the same sense,
 the former obviously being equivariant under the right action
 of~$G_0^{}$ as well.
\qed
\end{proof}

To complete the discussion, let us specify in what sense the map $D$
in equation~(\ref{eq:COVDER2}) captures the essence of the minimal
coupling prescription.
Abbreviating $P \times_{G_0} Q$ to~$E$, assume that $\, \varGamma:
M \longrightarrow CP \,$ is a section of~$CP$ representing a principal
connection in~$P$, $\varGamma^E: E \longrightarrow JE \,$ is the
section of~$JE$ (as a bundle over~$E$) representing the resulting
associated connection in~$E$, obtained by push-forward with the
canonical bundle map in equation~(\ref{eq:ASSCON1}), $\varphi:
M \longrightarrow E \,$ is a section of~$E$ and $\, \partial\varphi:
M \longrightarrow JE \,$ is its derivative (also denoted by $j\varphi$
and called its jet prolongation); then $\, D \,\smcirc\, (\varGamma,
\partial\varphi): M \longrightarrow \vec{J} E \,$ is indeed the covariant
derivative of~$\varphi$ with respect to that connection, because it is
elementary to see that equation~(\ref{eq:COVDER5}) combined with
equation~(\ref{eq:COVDER6}) will boil down to the formula in
equation~(\ref{eq:COVDER1}).

\subsection{Utiyama's theorem}

To deal with the curvature map, we observe that the bundle map $F$ of
equation~(\ref{eq:CURV2}) fits into the following commutative diagram
\begin{equation} \label{eq:CURV3}
 \begin{array}{c}
  \xymatrix{
   ~~\,\bar{J}^{\>\!2}_{\vphantom{p}} P\,~~
   \ar[r]^-{F} \ar[d]_-{J\rho_C^{}\,}
   & ~\pi_{JP}^*
      \Bigl( \rho^* \bigl(\bwedge^{\!2\,} T^* M \bigr) \otimes VP \Bigr)~
   \ar[d] \\
   ~J(CP)^{\vphantom{\hat{M}}}~ \ar[r]_-{F^{\vphantom{M}}}
   & \quad \bwedge^{\!2\,} T^* M \otimes
       (P \times_{G_0} \mathfrak{g}_0^{}) \quad
  }
 \end{array}
\end{equation}
where the bundles in the upper row are over $JP$ while those in the
lower row are over~$M$.
(Here, we have identified the vertical bundle $VP$ of~$P$ with the
trivial vector bundle $P \times \mathfrak{g}_0^{}$ over~$P$; then
the second tensor factor in the vertical map on the rhs of this diagram
is just the map $\rho_{\mathfrak{g}_0}^{}$ in the ``magical square'' of
equation~(\ref{eq:MSASSB1}) for the adjoint bundle $P \times_{G_0}
\mathfrak{g}_0^{}$, pulled back to~$JP$.)
Again, it is convenient to expand this to a commutative diagram
\begin{equation} \label{eq:CURV4}
 \begin{array}{c}
  \xymatrix{
   ~~\,\bar{J}^{\>\!2}_{\vphantom{p}} P\,~~
   \ar[r]^-{F} \ar[d]
   & ~\pi_{JP}^*
      \Bigl( \rho^* \bigl(\bwedge^{\!2\,} T^* M \bigr) \otimes VP \Bigr)~
   \ar[d] \\
   ~~\,\bar{J}^{\>\!2}_{\vphantom{p}} P\,~~
   \ar[r]^-{F} \ar[d]
   & \quad~~ \rho^* \bigl(\bwedge^{\!2\,} T^* M \bigr) \otimes VP ~~\quad
       \ar[d] \\
   ~J(CP)^{\vphantom{\hat{M}}}~ \ar[r]^-{F}
   & \quad \bwedge^{\!2\,} T^* M \otimes
       (P \times_{G_0} \mathfrak{g}_0^{}) \quad
  }
 \end{array}
\end{equation}
where the bundles in the middle row are over the total space~$P$
of the principal bundle.
(And again, we omit the labels on the vertical maps, which are
either the same as in the previous diagram or else are obvious.)
Then the bundle map~$F$ in the middle row is the alternator or
antisymmetrizer already introduced at the beginning of this paper
(see equation~(\ref{eq:ASPSSOJ})).
Continuing to use the same notation as in Section~3.2, we see that
this corresponds to the bundle map~$F$ in the bottom row being given
in terms of that in the top row according to
\begin{equation} \label{eq:CURV5}
 F_x^{} \bigl( J_{w_p}^{} \rho_C^{}
               (u'_{w_p}) \bigr)(v_1^{},v_2^{})~
 =~\rho_{\mathfrak{g_0}}^{}
   \bigl( F_{w_p}^{}(u'_{w_p})(v_1^{},v_2^{}) \bigr) \,,
\end{equation}
for $\, v_1^{},v_2^{} \in T_x^{} M$,
whereas the latter, as we recall from Section~2, is explicitly defined
as follows: given a semiholonomous second order jet $\, u'_{w_p} \in
\bar{J}_{w_p}^{\>\!2} P$, we arbitrarily choose some holonomous
second order jet $\, u_{w_p}^{\prime\,0} \in J_{w_p}^{\>\!2} P$
(this choice will ultimately drop out under the antisymmetrization)
to form the difference $u'_{w_p} - u_{w_p}^{\prime\,0}$,
which is a linear map from $T_x^{} M$ to the vertical space
$V_{w_p}^{\mathrm{jt}}(JP)$ of~$JP$ with respect to the jet
target projection $\pi_{JP}$; then we can apply the canonical
isomorphism
\begin{equation} \label{eq:CURV6}
 V_{w_p}^{\mathrm{jt}}(JP)~=~\ker T_{w_p}^{} \pi_{JP}^{}~
 =~T_{w_p}^{}(J_p^{} P)~
 \cong~\vec{J}_p^{} P~=~L(T_x^{} M,V_p^{} P)
\end{equation}
to identify it with a linear map from $T_x^{} M$ to $L(T_x^{} M,
V_p^{} P)$, that is, with an element of $L^2(T_x^{} M,V_p^{} P)$,
and obtain $\, F_{w_p}^{}(u'_{w_p}) \in L_a^2(T_x^{} M,V_p^{} P) \,$
by antisymmetrizing in the usual sense.
The last step then consists in applying the additional canonical isomorphism
\begin{equation} \label{eq:CURV7}
 V_p^{} P~\cong~\mathfrak{g}_0^{} \,.
\end{equation}
To show that $F_x^{}$ is well defined, note that if we replace the
point $p$ in~$P$ by another point in~$P$ in the same fiber over~$x$,
which is of the form $p \cdot g_0^{}$ for some (unique) $g_0^{} \in G$,
then we must replace $w_p^{}$ by $\, w_p^{} \cdot g_0^{} = T_p^{}
R_{g_0^{}}^P \,\smcirc\, w_p^{}$, $u'_{w_p}$ by $\, u'_{w_p} \cdot 
g_0^{} = T_{w_p}^{} R_{g_0^{}}^{JP} \,\smcirc\, u'_{w_p} \,$ and
similarly $u_{w_p}^{\prime\,0}$ by $\, u_{w_p}^{\prime\,0} \cdot
g_0^{} = T_{w_p}^{} R_{g_0^{}}^{JP} \,\smcirc\, u_{w_p}^{\prime\,0}$,
where $R_{g_0^{}}^P$ and $R_{g_0^{}}^{JP}$ denote right translation
by $g_0^{}$ in~$P$ and in~$JP$, respectively, so as to guarantee that
$J_{w_p}^{} \rho_C^{}(u'_{w_p})$ remains unaltered:
\[
\begin{aligned}
 J_{w_p \cdot g_0^{}} \rho_C^{}(u'_{w_p} \cdot g_0^{})~
 &=~T_{w_p \cdot g_0^{}} \rho_C^{} \,\smcirc\,
    T_{w_p}^{} R_{g_0^{}}^{JP} \,\smcirc\, u'_{w_p}~
  =~T_{w_p}^{}
    \bigl( \rho_C^{} \,\smcirc\, R_{g_0^{}}^{JP} \bigr)
    \,\smcirc\, u'_{w_p} \\[1ex]
 &=~T_{w_p}^{} \rho_C^{} \,\smcirc\, u'_{w_p}~
   =~J_{w_p}^{} \rho_C^{}(u'_{w_p}) \,.
\end{aligned}
\]
But then
\[
 u'_{w_p} \cdot g_0^{} \, - \, u_{w_p}^{\prime\,0} \cdot g_0^{}~
 =~T_{w_p}^{} R_{g_0^{}}^{JP} \,\smcirc\, u'_{w_p} \, - \,
    T_{w_p}^{} R_{g_0^{}}^{JP} \,\smcirc\, u_{w_p}^{\prime\,0} \,,
\]
so that applying the isomorphism in equation~(\ref{eq:CURV6}), we get
\[
 \bigl( u'_{w_p} - u_{w_p}^{\prime\,0} \bigr) \cdot g_0^{}~
 =~T_p^{} R_{g_0^{}}^P \,\smcirc\,
   \bigl( u'_{w_p} - u_{w_p}^{\prime\,0} \bigr) \,,
\]
and applying the additional isomorphism in equation~(\ref{eq:CURV7}),
we get
\[
 \bigl( u'_{w_p} - u_{w_p}^{\prime\,0} \bigr) \cdot g_0^{}~
 =~\mathrm{Ad}(g_0^{-1}) \,\smcirc\,
   \bigl( u'_{w_p} - u_{w_p}^{\prime\,0} \bigr) \,,
\]
implying that
\[
 [\, p \cdot g_0^{} \,,
     \bigl( u'_{w_p} - u_{w_p}^{\prime\,0} \bigr) \cdot g_0^{} \,]~
 =~[\, p \,,\, u'_{w_p} - u_{w_p}^{\prime\,0} \,] \,.
\]
(To justify this conclusion, note that the linear isomorphism
$\, T_{w_p}^{} R_{g_0^{}}^{JP}: T_{w_p}^{}(JP) \longrightarrow
T_{w_p \cdot g_0}^{}(JP)$, when restricted to the vertical space
of~$JP$ with respect to the jet target projection~$\pi_{JP}^{}$,
reduces to the tangent map $\, T_{w_p}^{} R_{g_0^{},p}^{JP}:
T_{w_p}^{}(J_p^{} P) \longrightarrow T_{w_p \cdot g_0}^{}
(J_{p \cdot g_0}^{} P) \,$ to the restricted right translation
$\, R_{g_0^{},p}^{JP}: J_p^{} P \longrightarrow
J_{p \cdot g_0}^{} P \,$ by $g_0^{}$.
But this is an affine map between affine spaces, so under the iso%
morphism in equation~(\ref{eq:CURV6}), its tangent map at each point
becomes the corresponding difference map, which is a linear map
$\, \vec{R}_{g_0^{},p}^{JP}: \vec{J}_p^{} P \longrightarrow
\vec{J}_{p \cdot g_0}^{} P$, and that is just composition with
\linebreak $T_p^{} R_{g_0^{}}^P: V_p^{} P \longrightarrow
V_{p \cdot g_0}^{} P$.
Finally, it is well known that under the isomorphism in equation~%
(\ref{eq:CURV7}), this becomes $\, \mathrm{Ad}(g_0^{-1}):
\mathfrak{g}_0^{} \longrightarrow \mathfrak{g}_0^{}$.)

Now we are ready to formulate the second main theorem in this paper,
which extends the left part of the commutative diagram in equation~(57)
of Ref.~\cite{FS}, as follows.
\begin{thm}~
 The curvature map $F$ in equation~(\ref{eq:CURV2}) is equivariant under
 the actions of the pertinent Lie groupoids, i.e., the diagram
 \begin{equation}
  \begin{array}{c}
   \xymatrix{
    \,\qquad\qquad\qquad J^2 G \times_M J(CP) \qquad\qquad\qquad\,~
    \ar[r] \ar[d]_{((\pi_{JG}^{\mathrm{fr}} \times \pi_{JG}^{}) \smcirc
                    \pi_{J^2 G},F)} &
    ~\,\quad\qquad J(CP) \qquad\quad\, \ar[d]^F \\
    (GL(TM) \times_M G) \times_M
    \bigl( \bwedge^{\!2\,} T^* M \otimes 
    (P \times_{G_0} \mathfrak{g}_0^{}) \bigr)~ 
    \ar[r] &
    ~\bwedge^{\!2\,} T^* M \otimes 
    (P \times_{G_0} \mathfrak{g}_0^{})
   }
  \end{array}
 \end{equation}
 commutes.
\end{thm}

\begin{proof}
 This follows immediately from Proposition~\ref{prp:EQUIV2}, together
 with the fact that, as shown in Section~3.1, the canonical isomorphism
 $\, VP \cong P \times \mathfrak{g}_0^{} \,$ and the projection
 $\, \rho_{\mathfrak{g}_0}^{}: P \times \mathfrak{g}_0^{}
 \longrightarrow P \times_{G_0} \mathfrak{g}_0^{} \,$ are
 both $G$-equivariant.
\qed
\end{proof}

To complete the discussion, let us specify in what sense the map $F$
in equation~(\ref{eq:CURV2}) captures the essence of the prescription
for defining the curvature of a principal connection.
Assume that $\, \varGamma: M \longrightarrow CP \,$ is a section
of~$CP$ representing a principal connection in~$P$ and $\, \partial
\varGamma: M \longrightarrow J(CP) \,$ is its derivative (also denoted
by $j\varGamma$ and called its jet prolongation); then $\, F \,\smcirc\,
\partial\varGamma: M \longrightarrow \bwedge^{\!2\,} T^* M \otimes
(P \times_{G_0} \mathfrak{g}_0^{}) \,$ is a $2$-form on~$M$ with
values in the adjoint bundle $P \times_{G_0} \mathfrak{g}_0^{}$
which is precisely the curvature form of that connection, because it
is elementary to see that equation~(\ref{eq:CURV5}) will boil down
to the formula in equation~(\ref{eq:CURV1}).

\section{Conclusions and Outlook}

The equivariance statements formulated in the two theorems in this
paper are very general, in that this equivariance holds for the full jet
groupoid $JG$ of the gauge groupoid~$G$, in the case of Theorem~1,
and for the full second order jet groupoid $J^{\>\!2} G$ of the
gauge groupoid~$G$, in the case of Theorem~2.
But this does of course not mean that a concrete field theoretical
model will have such a huge amount of symmetry~-- quite to the
contrary!
Any such model will be subject to restrictions on what are its allowed
symmetries coming from the dynamics, \linebreak which is governed,
say, by its Lagrangian: such a Lagrangian will typically be invariant
not under the pertinent jet groupoid but rather only under a certain
Lie subgroupoid thereof. \linebreak
The generic situation here, which prevails for all standard Lagrangians
in gauge theories, is that when $M$ comes equipped with some metric~%
$\mathslf{g}$, this Lie subgroupoid will be the inverse image of the
corresponding orthonormal frame groupoid $\, O(TM,\mathslf{g})
\subset GL(TM) \,$ under the ``frame'' projection from the pertinent
jet groupoid to the linear frame groupoid~$GL(TM)$ of~$M$. \linebreak
Thus what the two theorems in the previous section really prove is that
there are no other \linebreak restrictions, so this is in fact the correct
Lie groupoid for hosting the symmetries of any such theory, and remarkably,
it is large enough to accomodate not only its gauge symmetries but also
its space-time symmetries, including isometries as well as orthonormal
frame transformations, unifying them all within a single mathematical
object.
Finally, the formalism can also be adapted to handle symmetry breaking,
as has been discussed in Ref.~\cite{Key} (even though only at the level
of Lie group bundles and not of full Lie groupoids, which is however
enough to deal with that subject).

With this picture in mind, we hope to have demonstrated, in the two
papers of this series, that Lie groupoids provide a much wider and
more flexible mathematical framework than Lie groups for describing
symmetries in physics, and in some cases such as that of gauge
theories, we would venture to say they provide the ``right'' one.
What remains to be seen is how this approach will evolve when
one tries to extend it from classical to quantum field theories.

\section*{Acknowledgements}

The work of the first and the last author has originated from studies
performed during the development of their PhD theses, which have
been elaborated under the supervision of the second author and have
been supported by fellowships from CAPES (Coordena\c{c}\~ao de
Aperfei\c{c}oamento de Pessoal de N\'{\i}vel Superior), Brazil.
The second author would like to thank M.~Castrillon Lopez for
his hospitality during a visit to Madrid where part of this work
was developed, and for stimulating discussions about various
aspects of the subject.

\begin{appendix}

\section*{Appendix: Jet prolongations and gauge groupoids}

Our goal in this appendix is to prove a fact which is not used directly
in the main text (and that is why it has been relegated to an appendix)
but provides important additional insight into the way how Lie groupoid
theory is applied to gauge theories and has actually played a rather
important role in the development of the ideas underlying our work.
Briefly, the statement is that, (a) passing from a principal bundle first
to its jet prolongation and then to the gauge groupoid of that, or (b)
passing from a principal bundle first to its gauge groupoid and then
to the jet groupoid of that, gives the same result, up to a canonical
isomorphism; we may abbreviate this by saying that the processes
of building gauge groupoids and of taking jet prolongations commute,
provided the latter are interpreted correctly, each one in its category.
To show this, we must first explain the concept of jet prolongation
of a principal bundle.

\subsection*{Jet prolongations of principal bundles and associated bundles}

The main obstacle against an entirely trivial compatibilization between
the jet functor and the passage from principal bundles to associated
bundles resides in the fact that, although the (first order) jet bundle
of a fiber bundle is again a fiber bundle, the (first order) jet bundle
$JP$ of a principal bundle~$P$ is, by itself, \emph{not} a principal
bundle.
However, there is a simple way to remedy this defect, namely
by taking the fiber product with the linear frame bundle
$\mathrm{Fr}(M,GL(n,\mathbb{R}))$ of the base manifold.%
\footnote{Our notation for the linear frame bundle of a manifold
may at first sight look a bit clumsy, but it pays off by becoming
almost self-evident when we consider $G$-structures, which are
principal subbundles of the linear frame bundle with structure
groups that are closed subgroups $G$ of~$GL(n,\mathbb{R})$
and can with this notation simply be denoted by $\mathrm{Fr}
(M,G)$: a typical example would be the orthonormal frame
bundle $\mathrm{Fr}(M,O(n))$ induced by some Riemannian
metric.  (We apologize for the momentary change of meaning
of the symbol $G$ in this footnote).}
Indeed, it follows from the general constructions presented in
\cite[Chapter~4]{KMS} that if $P$ is a principal bundle over~$M$
with structure group~$G_0^{}$, then
\begin{equation}
 P^{(1)} \, = \, \mathrm{Fr}(M,GL(n,\mathbb{R})) \times_M JP
\end{equation}
is again a principal bundle over~$M$, called the (first order)
\emph{jet prolongation} of~$P$, with structure group
\begin{equation}
 G_0^{(1)} \, = \, (GL(n,\mathbb{R}) \times G_0^{}) \ltimes
                   L(\mathbb{R}^n,\mathfrak{g}_0^{})
\end{equation}
called the (first order) \emph{jet group} of~$G_0^{}$: this is simply the
semidirect product of the direct product $GL(n,\mathbb{R}) \times G_0^{}$
with the vector space $L(\mathbb{R}^n,\mathfrak{g}_0^{})$ of linear maps
from $\mathbb{R}^n$ to the Lie algebra~$\mathfrak{g}_0^{}$, which in this
context is viewed as an Abelian Lie group, where the semidirect product is
taken with respect to the natural (left) action
\begin{equation} \label{eq:JETGR1}
 \begin{array}{ccc}
  (GL(n,\mathbb{R}) \times G_0^{}) \times L(\mathbb{R}^n,\mathfrak{g}_0^{})
  & \longrightarrow & L(\mathbb{R}^n,\mathfrak{g}_0^{}) \\[1mm]
  ((a_0^{},g_0^{}) , \xi_0^{})
  &   \longmapsto   & (a_0^{},g_0^{}) \cdot \xi_0^{}
 \end{array}
\end{equation}
given by
\begin{equation} \label{eq:JETGR2}
 (a_0^{},g_0^{}) \cdot \xi_0^{} \,
 = \, \mathrm{Ad}(g_0^{}) \,\smcirc\, \xi_0^{} \,\smcirc\, a_0^{-1} \,,
\end{equation}
so the product in~$G_0^{(1)}$ is explicitly given by
\begin{equation} \label{eq:JETGR3}
 (a_{0,1}^{},g_{0,1}^{};\xi_{0,1}^{})
 (a_{0,2}^{},g_{0,2}^{};\xi_{0,2}^{}) \,
 = \, \bigl( a_{0,1}^{} a_{0,2}^{} , g_{0,1}^{} g_{0,2}^{};
             \xi_{0,1} + (a_{0,1}^{},g_{0,1}^{}) \cdot \xi_{0,2}^{} \bigr) \,.
\end{equation}
To write an explicit formula for the (right) action of $G_0^{(1)}$
on~$P^{(1)}$, we introduce the following notation: given any point
$p$ of~$P$, the isomorphism from the Lie algebra $\mathfrak{g}_0^{}$
onto the vertical space $V_p^{} P$ given by associating to every
$X_0^{}$ in $\mathfrak{g}_0^{}$ the value of the corresponding
fundamental vector field $(X_0^{})_P^{}$ at~$p$ (see equation~%
(\ref{eq:FVFP})) is, for any finite-dimensional real vector space~$W$,
extended to an isomorphism
\[
 \begin{array}{ccc}
  L(W,\mathfrak{g}_0^{}) \,\cong\, W^* \otimes \mathfrak{g}_0^{}
  & \longrightarrow & W^* \otimes V_p^{} P \,\cong\, L(W,V_p^{} P)
  \\[1mm]
  \xi_0^{}
  &   \longmapsto   & (\xi_0^{})_P^{}(p)
 \end{array}
\]
simply by taking the tensor product with the identity on~$W^*$
(i.e., $(\xi_0^{})_P^{}(p)(w) \equiv (\xi_0^{}(w))_P^{}(p)$).
Note that $G_0^{}$-equivariance of fundamental vector fields implies
that, denoting right translation by elements $g_0^{}$ of~$G_0^{}$
on~$P$ as well as on~$TP$ by $R_{g_0^{}}$ (so that $\, R_{g_0^{}}:
T_p^{} P \longrightarrow T_{p \cdot g_0}^{} P \,$ is the derivative
at~$p$ of $\, R_{g_0^{}}: P \longrightarrow P$), we have
\begin{equation} \label{eq:EQUIV1}
 R_{g_0^{}} \smcirc\, (\xi_0^{})_P^{}(p) \,
 = \, (\mathrm{Ad}(g_0^{})^{-1} \smcirc\,
      \xi_0^{})_P^{}(p \cdot g_0^{}) \,.
\end{equation}
Similarly, if we are given a (left) action
\[
 \begin{array}{ccc}
  G_0^{} \times Q & \longrightarrow &       Q        \\[1mm]
     (g_0^{},q)   &   \longmapsto   & g_0^{} \cdot q
 \end{array}
\]
of~$G_0^{}$ on some manifold~$Q$, then for any point $q$ of~$Q$,
we consider the linear map from the Lie algebra $\mathfrak{g}_0^{}$
into the tangent space $T_q^{} Q$ given by associating to every $X_0^{}$
in $\mathfrak{g}_0^{}$ the value of the corresponding fundamental
vector field $(X_0^{})_Q^{}$ at~$q$ (see equation~(\ref{eq:FVFQ}))
and, for any finite-dimensional real vector space~$W$, extend it to a
linear map
\[
 \begin{array}{ccc}
  L(W,\mathfrak{g}_0^{}) \,\cong\, W^* \otimes \mathfrak{g}_0^{}
  & \longrightarrow & W^* \otimes T_q^{} Q \,\cong\, L(W,T_q^{} Q)
  \\[1mm]
  \xi_0^{}
  &   \longmapsto   & (\xi_0^{})_Q^{}(q)
 \end{array}
\]
simply by taking the tensor product with the identity on~$W^*$
(i.e., $(\xi_0^{})_Q^{}(q)(w) \equiv (\xi_0^{}(w))_Q^{}(q)$).
Again, $G_0^{}$-equivariance of fundamental vector fields implies
that, denoting left translation by elements $g_0^{}$ of~$G_0^{}$
on~$Q$ as well as on~$TQ$ by $L_{g_0^{}}$ (so that $\, L_{g_0^{}}:
T_q^{} Q \longrightarrow T_{g_0 \cdot q}^{} Q \,$ is the derivative
at~$q$ of $\, L_{g_0^{}}: Q \longrightarrow Q$), we have
\begin{equation} \label{eq:EQUIV2}
 L_{g_0^{}} \smcirc\, (\xi_0^{})_Q^{}(q) \,
 = \, (\mathrm{Ad}(g_0^{}) \,\smcirc\,
      \xi_0^{})_Q^{}(g_0^{} \cdot q) \,.
\end{equation}
Then for $\, a_x^{} \in \mathrm{Fr}_x^{}(M,GL(n,\mathbb{R}))
= GL(\mathbb{R}^n,T_x^{} M)$, $u_p^{} \in J_p^{} P \subset
L(T_x^{} M,T_p^{} P)$, $a_0^{} \in GL(n,\mathbb{R})$,
$g_0^{} \in G_0^{} \,$ and $\, \xi_0^{} \in L(\mathbb{R}^n,
\mathfrak{g}_0^{})$,
\begin{equation} \label{eq:JETPR1}
 (a_x^{},u_p^{}) \cdot  (a_0^{},g_0^{};\xi_0^{}) \,
 = \, \bigl( a_x^{} \smcirc\, a_0^{} \,, R_{g_0^{}} \smcirc\,
             \bigl( u_p^{} + (\xi_0^{} \,\smcirc\, a_x^{-1})_P^{}(p) \bigr)
      \bigr) \,.
\end{equation}
Let us check explicitly that this formula does define a (right) action:
\begin{eqnarray*}
\lefteqn{\bigl( (a_x^{},u_p^{}) \cdot
                (a_{0,1}^{},g_{0,1}^{};\xi_{0,1}^{}) \bigr) \cdot
                (a_{0,2}^{},g_{0,2}^{};\xi_{0,2}^{})}
 \hspace*{1cm} \\[2mm]
 &=&\!\! \bigl( a_x^{} \smcirc\, a_{0,1}^{} \,,
                R_{g_{0,1}^{}} \smcirc\,
                ( u_p^{} + (\xi_{0,1}^{} \,\smcirc\, a_x^{-1})_P^{}(p) ) \bigr)
     \cdot (a_{0,2}^{},g_{0,2}^{};\xi_{0,2}^{}) \\[2mm]
 &=&\!\! \bigl( (a_x^{} \smcirc\, a_{0,1}^{}) \,\smcirc\, a_{0,2}^{} \,,
                R_{g_{0,2}^{}} \smcirc\, \bigl(
                R_{g_{0,1}^{}} \smcirc\,
                ( u_p^{} + (\xi_{0,1}^{} \,\smcirc\, a_x^{-1})_P^{}(p) ) \\
 & &\hspace{11em} \, + \,
                (\xi_{0,2}^{} \,\smcirc\, (a_x^{} \smcirc\, a_{0,1}^{})^{-1})_P^{}
                (p \cdot g_{0,1}^{}) \bigr) \bigr) \\[2mm]
 &=&\!\! \bigl( (a_x^{} \smcirc\, a_{0,1}^{}) \,\smcirc\, a_{0,2}^{} \,,
                R_{g_{0,2}^{}} \smcirc\, R_{g_{0,1}^{}} \smcirc\,
                \bigl( u_p^{} + (\xi_{0,1}^{} \,\smcirc\, a_x^{-1})_P^{}(p) \\
 & &\hspace{14em} \, + \,
                (\mathrm{Ad}(g_{0,1}^{}) \,\smcirc\, \xi_{0,2}^{} \,\smcirc\,
                 a_{0,1}^{-1} \,\smcirc\, a_x^{-1})_P^{}(p) \bigr) \bigr)
 \\[2mm]
 &=&\!\! \bigl( a_x^{} \smcirc\, (a_{0,1}^{} a_{0,2}^{}) \,,
                R_{g_{0,1}^{} g_{0,2}^{}} \smcirc\,
                \bigl( u_p^{} + ( ( \xi_{0,1}^{} + (a_{0,1}^{},g_{0,1}^{})
                        \cdot \xi_{0,2}^{} ) \,\smcirc\, a_x^{-1})_P^{}(p) \bigr)
                \bigr) \\[2mm]
 &=&\!\! (a_x^{},u_p^{}) \cdot
         \bigl( a_{0,1}^{} a_{0,2}^{} , g_{0,1}^{} g_{0,2}^{};
                \xi_{0,1} + (a_{0,1}^{},g_{0,1}^{}) \cdot \xi_{0,2}^{} \bigr)
 \\[2mm]
 &=&\!\! (a_x^{},u_p^{}) \cdot
         \bigl( (a_{0,1}^{},g_{0,1}^{};\xi_{0,1}^{})
                (a_{0,2}^{},g_{0,2}^{};\xi_{0,2}^{}) \bigr) \,.
\end{eqnarray*}

Further evidence that, in the case of principal bundles, the jet
prolongation in this sense~-- rather than just the usual jet bundle~--
is the correct object to consider can be accumulated by noting that
(a) the tangent bundle (of the total space), the jet bundle and the
linearized jet bundle of an associated bundle for~$P$ are all
associated bundles for~$P^{(1)}$ and (b) the connection
bundle $CP$ of~$P$ is also an associated bundle for~$P^{(1)}$,
i.e, there are canonical bundle isomorphisms
\begin{equation} \label{eq:TEASSB1}
 T(P \times_{G_0} Q) \, \cong \, P^{(1)} \times_{G_0^{(1)}}
 (\mathbb{R}^n \times TQ) \,,
\end{equation}
\begin{equation} \label{eq:JEASSB1}
 J(P \times_{G_0} Q) \, \cong \, P^{(1)} \times_{G_0^{(1)}}
 L(\mathbb{R}^n,TQ) \,,
\end{equation}
\begin{equation} \label{eq:LJASSB1}
 \vec{J} (P \times_{G_0} Q) \, \cong \, P^{(1)} \times_{G_0^{(1)}}
 L(\mathbb{R}^n,TQ) \,,
\end{equation}
and
\begin{equation} \label{eq:CPASSB1}
 CP \, \cong \, P^{(1)} \times_{G_0^{(1)}}
 L(\mathbb{R}^n,\mathfrak{g}_0^{}) \,,
\end{equation}
which preserve any invariant additional structures if such are
present (such as, for example, that of a vector bundle over~$P$
in the first and third case or that of an affine bundle over~$P$
in the second and fourth case).
Here, the relevant (left) actions of the structure group to be employed
in the definition of the associated bundles on the rhs of these equations
are
\begin{equation} \label{eq:TEASSB2}
 \begin{array}{ccc}
  G_0^{(1)} \times (\mathbb{R}^n \times TQ)
  & \longrightarrow & (\mathbb{R}^n \times TQ) \\[1mm]
  ((a_0^{},g_0^{};\xi_0^{}) , (v,v_q^{}))
  &   \longmapsto   & (a_0^{},g_0^{};\xi_0^{}) \cdot (v,v_q^{})
 \end{array}
\end{equation}
with
\begin{equation} \label{eq:TEASSB3}
 (a_0^{},g_0^{};\xi_0^{}) \cdot (v,v_q^{}) \,
 = \, \bigl( a_0^{} v \,, L_{g_0^{}}(v_q) -
                         (\xi_0^{}(a_0^{} v))_Q^{}(g_0^{} \cdot q) \bigr)
\end{equation}
in the first case,
\begin{equation} \label{eq:JEASSB2}
 \begin{array}{ccc}
  G_0^{(1)} \times L(\mathbb{R}^n,TQ)
  & \longrightarrow & L(\mathbb{R}^n,TQ) \\[1mm]
  ((a_0^{},g_0^{};\xi_0^{}) , u_q^{})
  &   \longmapsto   & (a_0^{},g_0^{};\xi_0^{}) \cdot u_q^{}
 \end{array}
\end{equation}
with
\begin{equation} \label{eq:JEASSB3}
 (a_0^{},g_0^{};\xi_0^{}) \cdot u_q^{} \,
 = \, L_{g_0^{}} \smcirc\, u_q^{} \,\smcirc\, a_0^{-1} -
      (\xi_0^{})_Q^{}(g_0^{} \cdot q)
\end{equation}
in the second case,
\begin{equation} \label{eq:LJASSB2}
 \begin{array}{ccc}
  G_0^{(1)} \times L(\mathbb{R}^n,TQ)
  & \longrightarrow & L(\mathbb{R}^n,TQ) \\[1mm]
  ((a_0^{},g_0^{};\xi_0^{}) , \vec{u}_q^{})
  &   \longmapsto   & (a_0^{},g_0^{};\xi_0^{}) \cdot \vec{u}_q^{}
 \end{array}
\end{equation}
with
\begin{equation} \label{eq:LJASSB3}
 (a_0^{},g_0^{};\xi_0^{}) \cdot \vec{u}_q^{} \,
 = \, L_{g_0^{}} \smcirc\, \vec{u}_q^{} \,\smcirc\, a_0^{-1}
\end{equation}
in the third case, and
\begin{equation} \label{eq:CPASSB2}
 \begin{array}{ccc}
  G_0^{(1)} \times L(\mathbb{R}^n,\mathfrak{g}_0^{})
  & \longrightarrow & L(\mathbb{R}^n,\mathfrak{g}_0^{}) \\[1mm]
  ((a_0^{},g_0^{};\xi_0^{}) , A_0^{})
  &   \longmapsto   & (a_0^{},g_0^{};\xi_0^{}) \cdot A_0^{}
 \end{array}
\end{equation}
with
\begin{equation} \label{eq:CPASSB3}
 (a_0^{},g_0^{};\xi_0^{}) \cdot A_0^{} \,
 = \, \mathrm{Ad}(g_0^{}) \,\smcirc\, A_0^{} \,\smcirc\, a_0^{-1}  + \xi_0^{}
\end{equation}
in the last case. 
(That these formulas do indeed define group actions follows by elementary
calculations, which we leave to the reader, using equations~(\ref{eq:JETGR2})%
-(\ref{eq:EQUIV2}).)

In order to explicitly construct the isomorphisms in equations~%
(\ref{eq:TEASSB1})-(\ref{eq:CPASSB1}), we resort to the ``magical
square'' for associated bundles, i.e., the commutative diagram in
equation~(\ref{eq:MSASSB1}), together with the commutative diagram
in equation~(\ref{eq:MSASSB3}) obtained by taking tangent maps and
the resulting quotient space representations for the tangent spaces
(see equation~(\ref{eq:TSASSB2})) and for the jet spaces (see
equation~(\ref{eq:JSASSB3})), plus a similar one for the linearized
jet spaces, to handle the first three cases, as well as to the
``magical square'' for the connection bundle, i.e., the commutative
 diagram in equation~(\ref{eq:MSCONB1}), to handle the last case.
More specifically, for the first case, consider the map
\begin{equation} \label{eq:TASSB1}
 \begin{array}{ccc}
  P^{(1)} \times (\mathbb{R}^n \times TQ)
  & \longrightarrow & TP \times TQ \\[1mm]
  ((a_x^{},u_p^{}) \,, (v,v_q^{}))
  &   \longmapsto   & \bigl( u_p^{} (a_x^{} v) \,, v_q^{} \bigr)
 \end{array}
\end{equation}
and observe that it takes
\begin{eqnarray*}
\lefteqn{\bigl( (a_x^{},u_p^{}) \cdot (a_0^{},g_0^{};\xi_0^{}) \,,
                  (a_0^{},g_0^{};\xi_0^{})^{-1} \cdot (v,v_q^{}) \bigr)}
 \hspace*{5mm} \\[1mm]
 &=&\!\! \bigl( (a_x^{},u_p^{}) \cdot (a_0^{},g_0^{};\xi_0^{}) \,,
                (a_0^{-1},g_0^{-1} ; - \mathrm{Ad}(g_0^{})^{-1} \,\smcirc\,
                                     \xi_0^{} \,\smcirc\, a_0^{}) \cdot (v,v_q^{})
                \bigr) \\[1mm]
 &=&\!\! \bigl( \bigl( a_x^{} \,\smcirc\, a_0^{} \,, R_{g_0^{}} \,\smcirc\,
                \bigl( u_p^{} + (\xi_0^{} \,\smcirc\, a_x^{-1})_P^{}(p) \bigr)
                \bigr) ,
                \bigl( a_0^{-1} v , L_{g_0^{-1}}(v_q^{}) +
                (\mathrm{Ad}(g_0^{})^{-1}(\xi_0^{}(v)))_Q^{}
                (g_0^{-1} \cdot q) \bigr) \bigr)
\end{eqnarray*}
which in the quotient $\, P^{(1)} \times_{G_0^{(1)}} (\mathbb{R}^n
\times TQ) \,$ represents the same class as $((a_x^{},u_p^{}) \,,
(v,v_q^{}))$, to
\begin{eqnarray*}
\lefteqn{\bigl( \bigl( R_{g_0^{}} \smcirc\,
                \bigl( u_p^{} + (\xi_0^{} \,\smcirc\, a_x^{-1})_P^{}(p) \bigr)
                \bigr) (a_x^{} v) \,,
                L_{g_0^{-1}}(v_q^{}) +
                (\mathrm{Ad}(g_0^{})^{-1}(\xi_0^{}(v)))_Q^{}
                (g_0^{-1} \cdot q) \bigr)}
 \hspace*{5mm} \\[1mm]
 &=&\!\! \bigl( R_{g_0^{}} (u_p^{}(a_x^{} v)) 
                 + (\mathrm{Ad}(g_0^{})^{-1}(\xi_0^{}(v)))_P^{}
                (p \cdot g_0^{}) ,
                L_{g_0^{-1}}(v_q^{}) +
                (\mathrm{Ad}(g_0^{})^{-1}(\xi_0^{}(v)))_Q^{}
                (g_0^{-1} \cdot q) \bigr)
\end{eqnarray*}
which in the quotient $\, T(P \times_{G_0} Q) \,$ represents the same class
as $(R_{g_0^{}} (u_p^{}(a_x^{} v)) \,,  L_{g_0^{-1}}(v_q^{}))$ since their
difference belongs to $\ker T_{(p \cdot g_0^{},g_0^{-1} \cdot q)} \rho_Q^{}$;
therefore, the map in equation~(\ref{eq:TASSB1}) induces a well defined map
\begin{equation} \label{eq:TASSB2}
 P^{(1)} \times_{G_0^{(1)}} (\mathbb{R}^n \times TQ)~~
 \longrightarrow~~T(P \times_{G_0} Q)
\end{equation}
between the quotient spaces which is the desired isomorphism.
Similarly, for the second case, consider the map
\begin{equation} \label{eq:JASSB1}
 \begin{array}{ccc}
  P^{(1)} \times L(\mathbb{R}^n,TQ)
  & \longrightarrow & JP \times L(TM,TQ) \\[1mm]
  ((a_x^{},u_p^{}) \,, u_q^{})
  &   \longmapsto   & \bigl( u_p^{} \,, u_q^{} \,\smcirc\, a_x^{-1} \bigr)
 \end{array}
\end{equation}
and observe that it takes
\begin{eqnarray*}
\lefteqn{\bigl( (a_x^{},u_p^{}) \cdot (a_0^{},g_0^{};\xi_0^{}) \,,
                  (a_0^{},g_0^{};\xi_0^{})^{-1} \cdot u_q^{} \bigr)}
 \hspace*{5mm} \\[1mm]
 &=&\!\! \bigl( (a_x^{},u_p^{}) \cdot (a_0^{},g_0^{};\xi_0^{}) \,,
                (a_0^{-1},g_0^{-1} ; - \mathrm{Ad}(g_0^{})^{-1} \,\smcirc\,
                                     \xi_0^{} \,\smcirc\, a_0^{}) \cdot u_q^{}
                \bigr) \\[1mm]
 &=&\!\! \bigl( \bigl( a_x^{} \smcirc\, a_0^{} \,, R_{g_0^{}} \smcirc\,
                \bigl( u_p^{} + (\xi_0^{} \,\smcirc\, a_x^{-1})_P^{}(p) \bigr)
                \bigr) , L_{g_0^{-1}} \,\smcirc\, u_q^{} \,\smcirc\, a_0^{} +
                (\mathrm{Ad}(g_0^{})^{-1} \,\smcirc\, \xi_0^{} \,\smcirc\,
                a_0^{})_Q^{} (g_0^{-1} \cdot q) \bigr)
\end{eqnarray*}
which in the quotient $\, P^{(1)} \times_{G_0^{(1)}} L(\mathbb{R}^n,TQ) \,$
represents the same class as $((a_x^{},u_p^{}) \,, u_q^{})$, to
\begin{eqnarray*}
\lefteqn{\bigl( R_{g_0^{}} \smcirc\,
                \bigl( u_p^{} + (\xi_0^{} \,\smcirc\, a_x^{-1})_P^{}(p) \bigr) ,
                L_{g_0^{-1}} \,\smcirc\, u_q^{} \,\smcirc\, a_x^{-1} +
                (\mathrm{Ad}(g_0^{})^{-1} \,\smcirc\, \xi_0^{}
                 \,\smcirc\, a_x^{-1})_Q^{} (g_0^{-1} \cdot q) \bigr)}
 \hspace*{5mm} \\[1mm]
 &=&\!\! \bigl( R_{g_0^{}} \smcirc\, u_p^{} +
                (\mathrm{Ad}(g_0^{})^{-1} \,\smcirc\, \xi_0^{} \,\smcirc\,
                 a_x^{-1})_P^{}(p \cdot g_0^{}) \,, \\
 & &\,       L_{g_0^{-1}} \,\smcirc\, u_q^{} \,\smcirc\, a_x^{-1} +
                (\mathrm{Ad}(g_0^{})^{-1} \,\smcirc\, \xi_0^{} \,\smcirc\,
                 a_x^{-1})_Q^{} (g_0^{-1} \cdot q) \bigr)
                \qquad\qquad\qquad\qquad
\end{eqnarray*}
which in the quotient $\, J(P \times_{G_0} Q) \,$ represents the same class
as $(R_{g_0^{}} \smcirc\, u_p^{} \,,  L_{g_0^{-1}} \,\smcirc\, u_q^{}
\,\smcirc\, a_x^{-1})$, since their difference belongs to $L(T_x^{} M,
\ker T_{(p \cdot g_0^{},g_0^{-1} \cdot q)} \rho_Q^{})$; therefore,
the map in equation~(\ref{eq:JASSB1}) induces a well defined map
\begin{equation} \label{eq:JASSB2}
 P^{(1)} \times_{G_0^{(1)}} L(\mathbb{R}^n,TQ)~~
 \longrightarrow~~J(P \times_{G_0} Q)
\end{equation}
between the quotient spaces which is the desired isomorphism.
For the third case, the argument is entirely analogous but somewhat simpler
since some terms drop out; we leave it to the reader to fill in the details.
Finally, for the last case, consider the map 
\begin{equation} \label{eq:JCONB1}
 \begin{array}{ccc}
  P^{(1)} \times L(\mathbb{R}^n,\mathfrak{g}_0^{})
  & \longrightarrow & JP \\[1mm]
  ((a_x^{},u_p^{}) \,, A_0^{})
  &   \longmapsto   & u_p^{} + (A_0^{} \,\smcirc\, a_x^{-1})_P^{}(p)
 \end{array}
\end{equation}
and observe that it takes
\begin{eqnarray*}
\lefteqn{\bigl( (a_x^{},u_p^{}) \cdot (a_0^{},g_0^{};\xi_0^{}) \,,
                  (a_0^{},g_0^{};\xi_0^{})^{-1} \cdot A_0^{} \bigr)}
 \hspace*{1cm} \\[1mm]
 &=&\!\! \bigl( (a_x^{},u_p^{}) \cdot (a_0^{},g_0^{};\xi_0^{}) \,,
                (a_0^{-1},g_0^{-1} ; - \mathrm{Ad}(g_0^{})^{-1} \,\smcirc\,
                                     \xi_0^{} \,\smcirc\, a_0^{}) \cdot A_0^{}
                \bigr) \\[1mm]
 &=&\!\! \bigl( \bigl( a_x^{} \smcirc\, a_0^{} \,, R_{g_0^{}} \smcirc\,
                \bigl( u_p^{} + (\xi_0^{} \,\smcirc\, a_x^{-1})_P^{}(p) \bigr)
                \bigr) , \mathrm{Ad}(g_0^{})^{-1} \,\smcirc\,
                (A_0^{} - \xi_0^{}) \,\smcirc\, a_0^{} \bigr)
\end{eqnarray*}
which in the quotient $\, P^{(1)} \times_{G_0^{(1)}} L(\mathbb{R}^n,
\mathfrak{g}_0^{}) \,$ represents the same class as $((a_x^{},u_p^{}) \,,
A_0^{})$, to
\begin{eqnarray*}
\lefteqn{R_{g_0^{}} \smcirc\, \bigl( u_p^{} + 
                                 (\xi_0^{} \,\smcirc\, a_x^{-1})_P^{}(p) \bigr) \,
         + \, (\mathrm{Ad}(g_0^{})^{-1} \,\smcirc\, (A_0^{} - \xi_0^{})
               \,\smcirc\, a_x^{-1})_P^{} (p \cdot g_0^{})}
 \hspace*{1cm} \\[1mm]
 &=&\!\! R_{g_0^{}} \smcirc\, \bigl( u_p^{} + 
                                 (\xi_0^{} \,\smcirc\, a_x^{-1})_P^{}(p) \bigr) \,
         + \, R_{g_0^{}} \smcirc\,  ((A_0^{} - \xi_0^{}) \,\smcirc\, a_x^{-1})_P^{} (p)
 \\[1mm]
 &=&\!\! R_{g_0^{}} \smcirc\, \bigl( u_p^{} +
                                 (A_0^{} \,\smcirc\, a_x^{-1})_P^{}(p) \bigr)
\end{eqnarray*}
which in the quotient $\, CP = JP/G_0^{} \,$ represents the same class
as $R_{g_0^{}} \,\smcirc\, u_p^{}$; therefore, the map in equation~%
(\ref{eq:JCONB1}) induces a well defined map
\begin{equation} \label{eq:JCONB2}
 P^{(1)} \times_{G_0^{(1)}} L(\mathbb{R}^n,\mathfrak{g}_0^{})~~
 \longrightarrow~~JP/G_0^{} = CP
\end{equation}
between the quotient spaces which is the desired isomorphism.

\pagebreak

Higher order jet prolongations can be constructed similarly, but
as in our previous work, we shall only use jet prolongations up
to second order, which can be constructed by iterating the first
order construction once and then performing an appropriate
reduction.
 
\subsection*{The jet groupoid of a gauge groupoid}

We begin with a more explicit description of the jet groupoid of
the gauge groupoid of a principal bundle~$P$, which is based on the
``magical square'' for gauge groupoids, i.e., the commutative diagram
\begin{equation} \label{eq:MSLGRP1}
\begin{array}{c}
\xymatrix{
 ~P \times P~ \ar[r]^-{\rho_P^{}}
 \ar@<-0.3em>[d]_-{\mathrm{pr}_2}
 \ar@<0.3em>[d]^-{\,\mathrm{pr}_1} & 
 ~(P \times P)/G_0^{}~
 \ar@<-0.3em>[d]_-{\sigma_G^{}}
 \ar@<0.3em>[d]^-{\,\tau_G^{}} \\
 \quad~ \vphantom{\hat{P}} P ~\quad \ar[r]_-{\rho^{\vphantom{m}}} &
 \qquad~ M \vphantom{\hat{M}} ~\qquad
}
\end{array}
\end{equation}
in which the horizontal projections define principal $G_0^{}$-bundles
while the vertical projections provide Lie groupoids (the first of which
is of course just the pair groupoid of~$P$) such that $\rho_P^{}$ is
an isomorphism on each (source or target or double) fiber.
Once more, we may seek to gain a more profound understanding of
the situation by extending this diagram to include the corresponding
jet groupoids, but a direct approach is not feasible here since we
cannot simply apply the jet functor to this diagram as we did before
(see equations~(\ref{eq:MSASSB2}) and~(\ref{eq:MSCONB2})),
the reason being that $P \times P$ is a Lie groupoid over~$P$
but not over~$M$.
Instead, we shall also consider the ``magical square'' for gauge
groupoids at the next level, which is the commutative diagram
\begin{equation} \label{eq:MSLGRP2}
\begin{array}{c}
\xymatrix{
 ~P^{(1)} \times P^{(1)}~ \ar[r]^-{\rho_{P^{(1)}}}
 \ar@<-0.3em>[d]_-{\mathrm{pr}_2}
 \ar@<0.3em>[d]^-{\,\mathrm{pr}_1} & 
 ~(P^{(1)} \times P^{(1)})/G_0^{(1)}~
 \ar@<-0.3em>[d]_-{\sigma_{G^{(1)}}}
 \ar@<0.3em>[d]^-{\,\tau_{G^{(1)}}} \\
 \qquad P^{(1)} \qquad \ar[r]_-{\rho^{(1)}} &
 \qquad\quad~~ M \vphantom{\hat{M}} ~~\quad\qquad
}
\end{array}
\end{equation}
and construct a canonical map
\begin{equation} \label{eq:JGGG1}
 P^{(1)} \times P^{(1)}~~\longrightarrow~~
 J((P \times P)/G_0^{})
\end{equation}
which we will show to be $G_0^{(1)}$-invariant, so it factors through
the projection $\rho_{P^{(1)}}$ to yield a canonical map
\begin{equation} \label{eq:JGGG2}
 \bigl( P^{(1)} \times P^{(1)} \bigr)/G_0^{(1)}~~\longrightarrow~~
 J((P \times P)/G_0^{})
\end{equation}
which will turn out to be an isomorphism (see Theorem~\ref{thm:JGGG} below).

To see how this construction goes, let us pick points $p_1,p_2 \in P$ with
$\, \rho(p_1) = x_1 \,$ and $\, \rho(p_2) = x_2$ and take tangent maps
to the commutative diagram in equation~(\ref{eq:MSLGRP1}) to obtain
the commutative diagrams
\begin{equation} \label{eq:MSLGRP2S}
 \begin{array}{c}
  \xymatrix{
   ~T_{p_2} P \oplus T_{p_1} P~
   \ar[rr]^-{T_{(p_2,p_{1\!})} \rho_P^{}}
   \ar[d]_{\mathrm{pr}_2}
   && ~T_{[p_2,p_{1\!}]} ((P \times P)/G_0^{})~
   \ar[d]^-{\,T_{[p_2,p_{1\!}]} \sigma_G^{}} \\
   \qquad \vphantom{\hat{P}} T_{p_1} P \qquad \ar[rr]_{T_{p_1} \rho}
   && \qquad\quad~ T_{x_1} M \vphantom{\hat{M}} ~\quad\qquad
  }
 \end{array}
\end{equation}
referring to the source projection and
\begin{equation} \label{eq:MSLGRP2T}
 \begin{array}{c}
  \xymatrix{
   ~T_{p_2} P \oplus T_{p_1} P~
   \ar[rr]^-{T_{(p_2,p_{1\!})} \rho_P^{}}
   \ar[d]_{\mathrm{pr}_1}
   && ~T_{[p_2,p_{1\!}]} ((P \times P)/G_0^{})~
   \ar[d]^-{\,T_{[p_2,p_{1\!}]} \tau_G^{}} \\
   \qquad \vphantom{\hat{P}} T_{p_2} P \qquad \ar[rr]_{T_{p_2} \rho}
   && \qquad\quad~T_{x_2} M \vphantom{\hat{M}} ~\quad\qquad
  }
 \end{array}
\end{equation}
referring to the target projection.
Since $\rho_P^{}$ is a submersion and hence its tangent maps are
surjective, this means that the tangent spaces $T_{[p_2,p_{1\!}]}
((P \times P)/G_0^{})$ of the orbit space $(P \times P)/G_0^{}$
can be realized as quotient spaces, namely, the linear maps
\begin{equation} \label{eq:TSLGRP1}
 T_{(p_2,p_{1\!})} \rho_P^{}: T_{p_2}^{} P \oplus T_{p_1}^{} P~~
 \longrightarrow~~T_{[p_2,p_{1\!}]}((P \times P)/G_0^{})
\end{equation}
induce isomorphisms
\begin{equation} \label{eq:TSLGRP2}
 T_{[p_2,p_{1\!}]}((P \times P)/G_0^{})~
 \cong~(T_{p_2}^{} P \oplus T_{p_1}^{} P)/
       \ker T_{(p_2,p_{1\!})} \rho_P^{} \,,
\end{equation}
with
\begin{equation} \label{eq:TSLGRP3}
 \ker T_{(p_2,p_{1\!})} \rho_P^{}~
 =~\{ ((X_0^{})_P^{}(p_2^{}),(X_0^{})_P^{}(p_1^{})) \, | \,
      X_0^{} \in \mathfrak{g}_0^{} \}~
 \cong~\mathfrak{g}_0^{} \,,
\vspace{1mm}
\end{equation}
where as before, $(X_0^{})_P^{}$ denotes the fundamental vector field
on~$P$ associated to a generator $\, X_0^{} \in \mathfrak{g}_0^{} \,$
via the pertinent action of~$G_0^{}$, as defined in equation~%
(\ref{eq:FVFP}) above.

With this notation, we can define the map in equation~(\ref{eq:JGGG1})
above, or more explicitly, its restriction to the fiber over the pair $(p_2^{},
p_1^{}) \in P \times P$, that is, the map
\begin{equation} \label{eq:JGGG3}
 \begin{array}{ccc}
  P_{p_2}^{(1)} \times P_{p_1}^{(1)}~
  =~J_{p_2} P \times GL(\mathbb{R}^n,T_{x_2} M) \times
    GL(\mathbb{R}^n,T_{x_1} M) \times J_{p_1} P
  & \longrightarrow & J_{[p_2,p_{1\!}]}((P \times P)/G_0^{})
  \\[1mm]
  (u_{p_2},a_{x_2},a_{x_1},u_{p_1})
  &   \longmapsto   & u_{[p_2,p_{1\!}]}
 \end{array}
\end{equation}
by setting
\begin{equation} \label{eq:JGGG4}
 u_{[p_2,p_{1\!}]}~
 =~T_{(p_2,p_{1\!})} \rho_P^{} \,\smcirc\,
   \bigl( u_{p_2} \,\smcirc\, a_{x_2} \,\smcirc\, a_{x_1}^{-1} \,,
           u_{p_1} \bigr) \,.
\end{equation}
We claim that this map is onto.
Indeed, any linear map $u_{[p_2,p_{1\!}]}$ from $T_{x_1} M$ to
$T_{[p_2,p_{1\!}]}((P \times P)/G_0^{})$ can be represented in
the form
\[
 u_{[p_2,p_{1\!}]}~
 =~T_{(p_2,p_{1\!})} \rho_P^{} \,\smcirc\,
   \bigl( \tilde{u}_{p_2},u_{p_1} \bigr)
\]
with linear maps $u_{p_1}$ from $T_{x_1} M$ to $T_{p_1} P$ and
$\tilde{u}_{p_2}$ from $T_{x_1} M$ to $T_{p_2} P$, where the pair
$(\tilde{u}_{p_2},u_{p_1})$ in $L(T_{x_1} M,T_{p_2} P \oplus
T_{p_1} P) \,$ is determined up to addition of a linear map
from $T_{x_1} M$ to~$\ker T_{(p_2,p_{1\!})} \rho_P^{}$.
Moreover, if we assume $u_{[p_2,p_{1\!}]}$ to be a jet in
$J_{[p_2,p_{1\!}]}((P \times P)/G_0^{})$ and to project to
some $\, a_{x_2,x_1} \in GL(T_{x_1} M,T_{x_2} M)$,
which we recall means that $\, T_{[p_2,p_{1\!}]} \,
\sigma_{(P \times P)/G_0}^{} \,\smcirc\, u_{[p_2,p_{1\!}]}
= \mathrm{id}_{\,T_{x_1} M}$ \linebreak while
$\, T_{[p_2,p_{1\!}]} \, \tau_{(P \times P)/G_0^{}}^{}
\,\smcirc\, u_{[p_2,p_{1\!}]} = a_{x_2,x_1}$, then we
conclude that $u_{p_1}$ will be a jet in $J_{p_1} P$ and
$\, u_{p_2} = \tilde{u}_{p_2} \,\smcirc\, a_{x_2,x_1}^{-1} \,$
will be a jet in $J_{p_2} P$.
Finally, we may write $\, a_{x_ 2,x_1}^{} = a_{x_2}^{} \,\smcirc\,
a_{x_1}^{-1} \,$ with $\, a_{x_1} \in GL(\mathbb{R}^n,T_{x_1} M) \,$
and $\, a_{x_2} \in GL(\mathbb{R}^n,T_{x_2} M)$.
It is then clear that the map in equation~(\ref{eq:JGGG3}) takes
\begin{eqnarray*}
\lefteqn{(u_{p_2},a_{x_2},a_{x_1},u_{p_1}^{}) \cdot
         (a_0^{},g_0^{};\xi_0^{})}
\hspace*{5mm} \\[1mm]
&=&\!\! \bigl( R_{g_0^{}} \smcirc\, \bigl( u_{p_2}^{} +
               (\xi_0^{} \,\smcirc\, a_{x_2}^{-1})_P^{}(p_2^{}) \bigr) \,,
               a_{x_2} \smcirc\, a_0^{} \,, a_{x_1} \smcirc\, a_0^{} \,,
               R_{g_0^{}} \smcirc\, \bigl( u_{p_1}^{} +
               (\xi_0^{} \,\smcirc\, a_{x_1}^{-1})_P^{}(p_1^{}) \bigr) \bigr)
\end{eqnarray*}
which in the quotient $\, (P^{(1)} \times P^{(1)})/G_0^{(1)} \,$
represents the same class as $(u_{p_2},a_{x_2},a_{x_1},u_{p_1})$, to
\[
 T_{(p_2 \cdot g_0,p_1 \cdot g_0)} \rho_P^{} \,\smcirc\,
 \bigl( R_{g_0^{}} \smcirc\, u_{p_2}^{} \,\smcirc\,
 a_{x_2}^{} \,\smcirc\, a_x^{-1} \,,
 R_{g_0^{}} \smcirc\, u_{p_1}^{} \bigr)
 = T_{(p_2,p_{1\!})} \rho_P^{} \,\smcirc\,
   \bigl( u_{p_2}^{} \,\smcirc\, a_{x_2}^{} \,\smcirc\, a_{x_1}^{-1} \,,
          u_{p_1}^{} \bigr)
\]
in $\, J_{[p_2 \cdot g_0,p_1 \cdot g_0]}((P \times P)/G_0^{})
= J_{[p_2,p_{1\!}]}((P \times P)/G_0^{})$.
This proves that the map in equation~(\ref{eq:JGGG2}) is well defined,
and it is now easy to see that it induces an isomorphism of Lie groupoids
over~$M$; we leave the details of the remainder of the proof to the
reader and just state the result as a
\begin{thm} \label{thm:JGGG}~
 Up to a canonical isomorphism, the (first order) jet groupoid of the
 gauge groupoid of a principal bundle~$P$ is equal to the gauge
 groupoid of its (first order) jet prolongation~$P^{(1)}$:
 \begin{equation} \label{eq:JGGG}
  J((P \times P)/G_0^{})~\cong~(P^{(1)} \times P^{(1)})/G_0^{(1)} \,.
 \end{equation}
\end{thm}

\end{appendix}

{\footnotesize

\end{document}